\documentclass[a4paper,12pt]{article}
\usepackage[utf8]{inputenc}

\usepackage{amssymb}
\usepackage{authblk}

\usepackage{setspace}
\usepackage{graphicx}
\usepackage{graphics}
\usepackage{fancyhdr}
\usepackage{amsmath}
\usepackage{amsfonts}
\usepackage{amssymb}
\usepackage{amsthm}
\usepackage{epsfig}
\usepackage{lipsum} %
\usepackage{color}
\usepackage{amsthm}
\usepackage{sectsty} % Allows customizing section commands
\usepackage{hyperref}
\usepackage{comment}
\usepackage{url}
\usepackage{empheq}
\usepackage{eurosym}
\usepackage{mathrsfs}
\usepackage[english]{babel}
\usepackage{booktabs}
\usepackage{caption}
\usepackage{subfig}
\usepackage{bm}
\usepackage{tikz} 
\usetikzlibrary{arrows,decorations.pathmorphing,backgrounds,fit,positioning,shapes.symbols,chains}

\usepackage{cancel}
\usepackage{xcolor}
\usepackage{color}
\usepackage{enumitem}
\usetikzlibrary{shapes,arrows}
\usepackage{listings}
\usepackage{rotating}

\usepackage[margin=1in]{geometry}

\theoremstyle{definition} \newtheorem{de}{Definition}[section]
\theoremstyle{plain}      \newtheorem{te}[de]{Theorem}
\theoremstyle{remark}     \newtheorem{os}[de]{Remark}
\theoremstyle{plain}      \newtheorem{pr}[de]{Proposition}
\theoremstyle{plain}      \newtheorem{lem}[de]{Lemma}
\theoremstyle{plain}	  \newtheorem{co}[de]{Corollary}
\theoremstyle{definition} 
\theoremstyle{remark}		
\theoremstyle{remark}		
\theoremstyle{plain}        \newtheorem*{beware*}{Beware}
\theoremstyle{definition} \newtheorem{as}{Assumption}

\usepackage[round,sort&compress,comma]{natbib}

\usepackage[titletoc]{appendix}
\usepackage{titlesec}
\titlelabel{\thetitle\quad}

\usepackage{multirow,array}

    \usepackage{color}
\usepackage{tikz}

\usetikzlibrary{shapes,decorations,arrows,calc,arrows.meta,fit,positioning}
\tikzset{
    -Latex,auto,node distance =1 cm and 1 cm,semithick,
    state/.style ={ellipse, draw, minimum width = 0.7 cm},
    point/.style = {circle, draw, inner sep=0.04cm,fill,node contents={}},
    bidirected/.style={Latex-Latex,dashed},
    el/.style = {inner sep=2pt, align=left, sloped}
}

\title{
Multi-Asset Bubbles Equilibrium Price Dynamics}
\author[1,\footnote{
Corresponding author.}]{Francesco Cordoni}
\affil[1]{\small
Department of Economics, Royal Holloway University of London, Egham TW20 0EX, UK.
\protect\\ E-mail: fc.cordoni@gmail.com
\vspace{5pt}} 

\date{
  \today}
\begin{document}

% FOR THE PROOF, ADD AUTHORS DETAIL, FIX TYPO CONSISTENLTY A PAGE 31, ADD ACKNOLEDGEMENTS
\maketitle

% spiegare meglio cosa succede qunado il prezzo non è all'equilibrio, che si crea uno spread

%intro
%market design experiment with two assets
%agent based model review Duffy unver  con giustificazione
% baghestian  (senza spiegare nel dettaglio)
%risultati duffy unver
%risultati per one asset
%extension on the two asset case with factor investing 
%--fine duffy unver
%how to drop probability and weak foresitgh assumption
%(spiegazione modello in two asset case
%baghestian model equilibrium 
%extension in two asset with factor investing
%identification problem
%solving the identification problem

\begin{abstract}
The price-bubble and crash formation process is theoretically investigated in a two-asset equilibrium model.
  Sufficient and necessary conditions are derived for the existence of average equilibrium price dynamics of different agent-based models, where agents are distinguished in terms of factor and investment trading strategies. 
  In line with experimental results,
we show that assets with a positive average dividend, i.e., with a strictly declining fundamental value,
display at the equilibrium price the typical hump-shaped bubble observed in experimental asset markets.
Moreover, a misvaluation effect is observed in the asset with a constant fundamental value,  triggered by the other asset that displays the price bubble shape when a sharp price decline is exhibited at the end of the market.
%   We found that the equilibrium price is obtained when traders apply the same investment strategies for both assets.
 
\end{abstract}


% \section*{Acknowledgements}

\clearpage
\section{Introduction}

Financial bubbles cast doubt over
agents' rationality and represent
possible sources of inefficiencies and market fragilities.
 However, due to the several underlying mechanisms that can lead to their formation, a clear understanding of the origins of the price bubble is still missing. 
 Market restrictions are interestingly related to explaining the price bubble phenomena. Indeed, in a market where short-selling is not allowed, the price dynamics can 
 be raised by ``excessively optimistic" traders making  the market generally overvalued, \citep{miller1977risk}.  
 Moreover, market liquidity plays a relevant role 
to produce spillovers contagion effect, which leads to (flash) crashes events (\citealt{CFTC}, \citealt{cespa2014illiquidity}, \citealt{kirilenko2017flash},
\citealt{da2017investors}).
 On top of that, price bubbles can be attributed to agent's confusion, lack of rationality, speculation and other factors, see e.g., \cite{smith1988bubbles}, \cite{haruvy2006effect}, \cite{kirchler2012thar}, \cite{baghestanian2015traders}.
 
Since the seminal paper of  \citet{smith1988bubbles}
(henceforth, SSW), experimental asset markets proved to be a powerful tool for analyzing bubble-crash patterns and
the agents' behavioural strategies through laboratory market experiments.
Indeed, these price bubbles and crashes are 
robust and persistent under different experimental laboratory settings.
For instance,
\cite{kirchler2012thar} explored where agents' confusion about fundamental value combined with ample liquidity can lead to
significant mispricing and overvaluation and so increasing the price bubble-shape pattern\footnote{See \cite{palan2009bubbles, palan2013review} for an exhaustive review.}.

However,
rational bubble theories, e.g., \cite{blanchard1979speculative}, \cite{tirole1985asset}, \cite{froot1991intrinsic}, \cite{abreu2003bubbles}, \cite{MartinAER}, \cite{Hirano2016}, \cite{MiaoAER} provide little help to properly understand laboratory asset bubble phenomena\footnote{For instance, the theoretical 
literature has shown that bubbles may arise in an infinite horizon setting or with an infinite number of traders, while 
laboratory markets exhibit price bubbles in finite horizons and the number of traders is finite, see also the discussion in \cite{duffy2006asset}.}.
This difficulty in linking existing price formation theories to laboratory asset markets emphasizes the puzzling
feature of these experimental price bubbles, \citep{smith2000dividend}.
Thus, in the spirit of \cite{duffy2006asset}, \cite{haruvy2006effect}, \cite{caginalp2008dynamics}, and \cite{baghestanian2015traders}, we aim to theoretically investigate the price-bubble mechanism employing agent-based approaches, instead of conducting additional
experiments.

In this work, sufficient and necessary conditions are presented for the existence of average equilibrium price dynamics of various agent-based models in a two-asset market. Specifically, one asset has a declining fundamental value, named  \emph{speculative} asset, while the other has a constant fundamental value, referred to as \emph{value} asset.
Starting from the single-asset \cite{duffy2006asset} (henceforth, DU)  model, we show
how to recover and analyze the price formation process finding the related equilibrium and average price dynamics expressions.
We then extend the DU model to the two-asset case, presenting
different factor trading strategies characterizing the equilibrium prices in the two assets market.
Then, following \cite{baghestanian2015traders}, 
we introduce heterogeneous agents with short horizon investment strategies, allowing us to relax the use of exogenous probability of the standard DU model to decide whether a trader is a buyer or a seller.
For those assets with a positive average dividend, i.e., speculative assets,
the typical price hump-shaped of price bubble observed in experimental asset markets is displayed at equilibrium price dynamics.
% Moreover, consistently with the experiments of \cite{kirchler2012thar}, we found that the value asset also maintains its equilibrium price dynamics constant.

% {\color{blue}
Although the original DU model implemented a double-auction market mechanism, we have opted for a single Walrasian auction system.
This decision allows us to offer a more robust theoretical justification for our equilibrium concept. The rationale behind this change stems from the recent findings presented by \cite{jantschgi2022double}. Starting from a unified definition of the double auction mechanism, applicable to both finite and infinite markets, and without making any assumptions of regularity for the supply and demand function, the authors have shown that double-auction can implement market clearance and Walrasian equilibrium. Furthermore, they demonstrate that this unified definition includes, as special cases, the k-double auction mechanism for finite markets and models of continuous and strictly monotone demand and supply for infinite markets, establishing the convergence of finite auctions to infinite auctions. In addition, \cite{ikica2023competitive}
presented findings from a series of controlled continuous double auction experiments aimed at reproducing and stress-testing the phenomenon of convergence to competitive equilibrium\footnote{The double-auction system could be considered reminiscent of the Walrasian auction where the notion of competitive equilibrium originates, \cite{ikica2023competitive}.} under private information. They found that convergence to competitive equilibrium occurs after a handful of trading periods, despite an initial asymmetry favoring buyers, which results in prices below equilibrium levels.

The DU model is one of
the first model to employ an agent-based computational approach with noise (``near-zero-intelligence") traders  to study the sources of bubble-crash patterns in a single-asset market, replicating the experiments of 
SSW\footnote{ \cite{haruvy2006effect} also have shown that
similar patterns of experimental markets are also generated by
simulated markets with heterogeneous agents, e.g., fundamentalist, speculators, and feedback agents.}.
\cite{duffy2006asset} have followed the methodology of
\cite{gode1993allocative,gode1994human} to explore the role of ``zero-intelligence" machine traders in experimental markets, comparing the results from the artificial traders market with that of human traders.

% In a trading round, the traders of the DU model have to decide  two things: the position, i.e., they have to choose to be sellers or buyers, and the quote, i.e., the price they are willing to pay for selling or buying.
 In a trading round, the traders of the DU model have to decide two things: the position, i.e., they have to choose whether to sell or buy, and the quote, i.e., the price they are willing to pay for selling or buying.
The traders' positions are decided in each trading round by a Bernoulli variable, where the probability of being a buyer decreases in each round, so that in the last trading sessions, traders are more prone to sell. This last condition is called \emph{weak-foresight assumption}.
Once decided the position, traders place orders following a weighted average between the previous period prices and a (random) value proportional to the fundamental value, which incorporates traders' confusion about the fundamental value. The weighting parameter is called \emph{anchoring} parameter, where a high value indicates that traders are more likely to post quotes close to the previous period prices.
 The anchoring effect captures the behavioural notion that anchoring might be relevant to explaining price-bubble shape, because it causes transaction prices to start low and subsequently rise as trade proceeds, see \cite{duffy2006asset} and \cite{baghestanian2015traders}.

 Despite its simplicity, the DU model explains some of the underlying mechanisms of price bubbles through the agent-based model approach. For instance, 
 when the fundamental value of the asset decreases over time, e.g., the average dividend of the asset is positive, agents
 start trading the stock at a low value compared to the fundamental one due to inexperience. Then, traders gain confidence to create an upward trend, with a subsequent soaring of the price dynamics. Agents
 will post quotes at a high level compared to fundamental values due also to their confusion about the fundamental value of the asset. This confusion is incorporated in the DU model by the underlying randomness of traders'  bid quotes. Then as the last trading rounds approach,
large-scale selling orders are posted by traders since it decreases their subjectively perceived probability of being able to sell. This induced mechanism is modelled by DU employing the above mentioned \emph{weak-foresight} assumption.

 On the other hand, the DU model setting contains some simplifications that make their model far from the real market setting and limit their results to the experimental context. 
 For instance, 
the trading on one
asset can trigger price changes on other assets, and, as seen during the Flash Crash of 2010,
instability can influence a large set of assets, \cite{CFTC}.
The execution of asset portfolio orders, and more generally, the commonality in liquidity across assets, \cite{chordia2000commonality}, \cite{tsoukalas2019dynamic}, may cause price changes among assets and due to cross-impact effects trigger significant instabilities effects across all market segments, \cite{cordoni2020instabilities}. 
Therefore, in a multi-asset market, can the 
 price bubble of one asset propagate to all the other assets? How would this propagation be characterized?
 Can spillover effects or specific (factorial) trading strategies, triggered by the bubble of one asset, also affect other assets' price dynamics?

Interestingly, \cite{caginalp2002speculative}
partially explored the above questions through experiments, where
the presence of price-bubbles tends to increase volatility and diminishes the prices of other stocks.
\cite{fisher2000experimental} also conducted a similar two-assets market experiment. They investigated the exchange rate dynamics between two assets, reporting that this rate converged quickly to its theoretical value.
\cite{Ackert2006art, ackert2006origins}, have also investigated experiments with two assets, analyzing the effects of margin buying and short-selling where one of the asset is a lottery asset.
Furthermore, \cite{oechssler2007asset} performed experimental markets where five different assets can be traded simultaneously.

However, to the best of our knowledge, little attention has been given to the study of multi-assets experimental markets employing agent-based modelling approaches to investigate the price-bubble mechanism. 
A recent further extension in a two-asset market of the DU model was proposed by \cite{cordoni_2021_simulation}, where the role of market impact was investigated in the price bubble formation.  
In particular, in \cite{cordoni_2021_simulation} each agent is designed in order to follow different factor-investing style strategies, where traders decide to buy or sell assets depending on the factor they have chosen. 
The authors found evidence that the liquidity mechanism which generates the price bubble does not involve a symmetric cross-impact between the two assets, i.e., they found that the price changes in one asset is caused by the trading on other assets.

We present different factor trading strategies characterizing the equilibrium prices on the two asset extensions.
When traders adopt one of these factor trading strategies, the average price dynamics of one of the assets reveals a misvaluation.
The difference between the average price and the fundamental value of this asset arises from the supply and demand imbalance generated by traders at the end of the market session, resulting in a ``contagion" effect between the price dynamics of the two assets.
We investigate the conditions under which this ``contagion'' effect occurs depending on the factor chosen by agents.

% When traders follow a factor trading strategy,
% the average price dynamics of the value asset will display a misvaluation. This difference between  average prices and fundamental value of the value asset is triggered by the supply and demand imbalance generated by traders at the end of the market session when the speculative asset price-bubble declines.
% We investigate under which conditions this ``contagion" effect of the price-bubble shape on the value asset, triggered by the sharp decline at the end of market periods,  
% depending on the factor chosen by agents. 

Another sticking point of the DU model is the use of an exogenous probability parameter by traders to decide whether to buy or sell an asset. 
% In other words, the DU model is a purely statistical model which does not rely on any economic interpretation.
Therefore, recently, their model has been generalized\footnote{Even if in a call-market trading environment, while the original work of \cite{duffy2006asset} was developed for continuous double-auction markets as in \cite{smith1988bubbles}.} by \cite{baghestanian2015traders} (in a single asset market),
by introducing heterogeneous agents, which use fundamentalist and speculative 
short horizon investment strategies together with noise traders. 
We combine the fundamentalists and speculators investment strategies with factor-investing style strategies, highlighting how an identification issue arises in the two-assets market equilibrium.
Specifically, different market settings, which depend on the market factors chosen by agents,  generate the same equilibrium price dynamics, confounding the origin and motivation of the average price-bubble dynamics.
However, we identify the factor strategy characterizing the two-asset price equilibrium by extending the fundamentalist and speculative investment strategies to the two-assets case.
% Fundamentalist traders employ their information on fundamental value to track the current price and decide to buy or sell depending on whether the current price is below or above the fundamental value. 
% On the other hand, speculators decide whether buy or sell depending on their expectation of future clearing prices formed at the beginning  of trading sessions.
% Thus, we interpret  the \cite{baghestanian2015traders} results on the 
% the \emph{boom} and \emph{bursts} phase of price bubble dynamics, 
%  using our characterization of supply/demand imbalance.
%At the end of trading sessions, speculator
% and fundamentalist traders place sell orders, supported by the demand of noise traders, which cause a subsequent price decline, i.e., the bursts of the bubble, see \cite{baghestanian2015traders}. 

In Section \ref{sec_market_setup} we introduce notation and our market setting. In Section \ref{sec_DU_model_start} we recall the DU model and  the corresponding equilibrium price is derived.
In Section \ref{sec:two_assets} and \ref{sec_hetero_theory}
we present our main results to the two-asset case using
heterogeneous agent based model with factor and investment strategies, respectively. Finally, in Section \ref{sec_conc}
we conclude.
% %qui
% by employing standard as the Roll model
 
%commento veloce su duffy unver- baghestina-> necessità di multi asset environment-> risultato nostro

%arrivato qui con spell check
 \section{Market Setup}\label{sec_market_setup}

To investigate the price bubble and crash mechanism and the related price dynamics
 for a multi-assets market environment, we set up 
 a market composed of two assets.
 The first asset has a positive average dividend $\overline{d}_1>0$, while the second asset average dividend is null, $\overline{d}_2=0$.
 Therefore, the two assets have different fundamental value dynamics; the fundamental value of the first asset, $FV_1$, decreases over time, while for the second one, $FV_2$, is constant.
%  We may select the terminal value for the two assets so that the two fundamental values intersect at the middle of the trading session.
 Unless specified, we follow the specifications presented in 
 \cite{duffy2006asset} and \cite{cordoni_2021_simulation} by setting the dividend distribution support
  of asset $1$ equal to
$\{\$0, \$0.1, \$0.16, \$0.22\}$ and terminal (buy-out) value  $TV_1=\$1.80$, and for asset $2$, $\{\$-0.2, \$-0.1, \$0, \$0.1, \$0.2\}$ and terminal value $TV_2=\$2.80$, for asset 2, where a negative dividend corresponds to an holding cost, see \cite{kirchler2012thar}. 
% Therefore, asset $P_1$ has a decreasing fundamental value, while $P_2$ has a constant fundamental value equal to its 
% terminal value.

 In the following, we investigate the existence of equilibrium price dynamics for the previous two-assets market. We first focus on a single-asset market composed only of the first asset and then generalize our results in the two-assets case. 

\section{The Duffy-\"{U}nver Agent-Based model}
\label{sec_DU_model_start}

% To investigate the typical bubble crash mechanism and derive our hypotheses
% for the two-assets market experiment, in this section we set-up a
% series of agent-based models similar to \cite{duffy2006asset} and
% \cite{baghestanian2015traders}. The \cite{duffy2006asset} is a model
% with near-zero-intelligence agents designed to generate and replicate
% laboratory the market bubbles following as in \cite{smith1988bubbles},
% while \cite{baghestanian2015traders} is a model of heterogeneous
% agents, including both sophisticated\emph{, }i.e. \emph{fundamentalist}
% and \emph{speculative} traders (who act strategically), and \emph{noise}
% agents (who place orders randomly as in \cite{duffy2006asset}). We
% extend both of these models to a multi-asset environment. In particular,
% we start from the near-zero-intelligence agents based model of \cite{duffy2006asset},
% to then specify traders' strategies as in of \cite{baghestanian2015traders}.
% Finally, we further extend this set-up to additionally provide a model
% that considers the existence of market-markers, as in our design.

The DU model involves $N$ agents who trade the same asset  in $T$ trading periods. 
A random dividend is paid at the end of each trading period $t$.
Then,
the (average) fundamental value  is given by 
\[
FV_{t,1}=(T-t+1) \overline{d}_{1}+TV_{1},
\]
where $\overline{d}_{1}$ is the expected dividend payment, 
and $TV_{1}$ represents the terminal value.
The dividends are drawn by
a uniformly distributed random variable with finite support, while 
the terminal value is fixed to a constant value, (see Section \ref{sec_market_setup}).
For the sake of simplicity, in this section, we omit the subscript $1$, since we focus only on the one-asset case.

 In the original work of \cite{duffy2006asset} traders can post bid/ask quotes during submission rounds in trading time interval $t$. Precisely,
 each trading period
$t$ is composed of $S$ submission rounds, where 
traders can place
their orders following a double auction market mechanism with continuous open-order book dynamics.
However, since we focus on the average equilibrium price, for our analysis we can omit this  submission rounds architecture from the trading model\footnote{We may relate our analysis to batch trading markets, where orders are first accumulated and then executed simultaneously at the equilibrium price, which clears demand and supply.}.
% for our analysis we can omit this  submission rounds architecture from the trading model\footnote{We may relate our analysis to batch trading markets, where orders are first accumulated and then executed simultaneously at the equilibrium price, which clears demand and supply.}.
% Indeed, \cite{duffy2006asset} introduce submission round to replicate typical order book dynamics observed by experiments, which is unnecessary for our study.

% \subsubsection{Basic Model specification}

At the beginning of market session,
each trader $j$ has an endowment
of cash $x^{j}$ and a quantity of the asset $y^{j}$.
All agents are equally informed about the fundamental value dynamics.
At trading period $t$ an agent $j$ is
a buyer with probability $\pi_{t}$ and a seller with probability
$1-\pi_{t}$, where it is assumed the so-called \emph{weak foresight} assumptions, i.e., the probability of being a buyer
is decreasing across the trading periods, 
\[
\pi_{t}=\max\{0.5-\varphi t,0\},\text{ where \ensuremath{\varphi\in\left[0,\frac{0.5}{T}\right)}. }
\]
A positive $\varphi$
implies a gradual increase of excess supply towards the end of the
market session and so it contributes to the reduction in mean transaction
prices. 
This assumption makes the DU model results quite
consistent with the experimental data of \cite{smith1988bubbles}, where also a decline in average transaction volume is observed across trading rounds. 
We discuss in detail the effect of this assumption in our equilibrium analysis. 
Each quote submitted by both a seller or a buyer is for one asset share.
A buyer $j$ in period $t$ can place
a bid quote if enough cash balances $x_{t}^{j}>0$ 
is available in his account. On the other hand, 
sellers can place an ask quote if they have 
enough share quantity,  $y_{t}^{j}>0$. 
Thus, agent
$j$ places a quote which is provided by a convex combination of the previous period traded price, $\overline{p}_{t-1}$, and a random quantity $u_{t}$.
This random variable $u_t$ captures the  uncertainty about 
 agents' decisions and it
has a distribution with support $[0,\kappa\cdot FV_{t}]$,
where $\kappa>0$. If not specified, $\kappa$ is assumed to be greater than 1. This randomness was introduced by \cite{duffy2006asset} to capture agent's confusion on the fundamental value. 
 In the original work of DU,  the distribution of the random variable $u_t$  was uniform. However, since we are interested in modelling the average submitted quote dynamics, we consider a distribution for which the first moment is finite, it serves as a good representative statistic of central tendency measure, like for the Gaussian or triangular distribution, and it is centred to the middle of the support, i.e., it is equal to $\frac{\kappa}{2}\cdot FV_{t}$.

At time $t=1$ \cite{duffy2006asset} set $\overline{p}_0=0$ in order to exactly replicate the same shape exhibited by SSW experiments.
Specifically, this condition ensures that the price-bubble will start at a value below the fundamental value. This phenomenon in the SSW experiment results from the participants' inexperience, and it induces an upward trend when agents gain confidence in adjusting the price to the fundamental value. Therefore, this condition artificially triggers the price-bubble mechanism, and we decide to set  $\overline{p}_{0}$
equal to the fundamental value at time $t=1$, contrary to the DU model.
Furthermore, the assumption of inexperienced participants is far from the real financial market, where traders are highly specialised due to the increase in market competitiveness.
Therefore, our condition enables us to study the bubble mechanism in a complementary way with respect to \cite{duffy2006asset} analysis, since 
 our zero-intelligence agents are assumed to be sufficiently more experienced than those of DU model. Moreover, this assumption is also in line with the recent experiments discussed in \cite{baghestanian2015traders}, where the price-bubble starts close to the fundamental value.
 
Therefore, if $j$ is a buyer,  the trader will place
a bid price, given by 
\[
b_{t}^{j}=\min\{(1-\alpha)u_{t}^j+\alpha\overline{p}_{t-1},x_{t}^{j}\}, \quad \alpha\in (0,1)
\]
where $u_{t}^j$ denotes\footnote{For the sake of notation simplicity, and since we will study average price dynamics, in the following we will omit to specify the superscript $j$ to the random variable $u_t$ when it is not necessary.} the realization of the random variable $u_t$ for the $j$-th agent,
and if $j$ is a seller, the agent will place, if $j$ has at least one share, an ask price given by 
\[
a_{t}^{j}=(1-\alpha)u_{t}^j+\alpha\overline{p}_{t-1}.
\]
The parameter $\alpha\in(0,1)$ is called \emph{anchoring} parameter and it represents agent's attitude to post quotes close to previous period price. As observed by \cite{duffy2006asset} quotes converge on average to $\kappa\frac{FV_{t}}{2}$.

The anchoring parameter plays a crucial role in the price-bubble formation in the DU model.
Prices will necessarily increase initially and decrease
as the fundamental value decreases.
\cite{duffy2006asset} argued that this kind of explanation
for the price-bubble mechanism holds regardless of $\varphi$.
However, when $\varphi=0$, the price will continue to get a ``hump-shaped" path with no decrease in transaction volume.
We completely characterize the equilibrium price dynamics in function of the above parameters in Section \ref{sec_DU_eq_price}. 
The standard DU model agents are often referred to as near-zero-intelligence traders due to the simple trading strategies they implement and in the DU model extensions they are associated to noise traders strategies, see \cite{baghestanian2015traders} and \cite{cordoni_2021_simulation}.

We observe that we can adopt the unified definition of the double auction mechanism of \cite{jantschgi2022double} to show that a double auction mechanism can implement market clearance and Walrasian equilibrium.
Additionally, \cite{ikica2023competitive} have replicated this convergence to competitive equilibrium from a series of controlled continuous double auction experiments, observing convergence after few trading periods. 
Therefore, without loss of generality, replacing the repeated double auction system of the original DU model with a single period Walrasian auction could help to clarify the definition of equilibrium we aim to study.
As a result, we opt to substitute the repeated double auction system with a single period Walrasian auction.

\subsection{Equilibrium Average Price Dynamics  with Homogeneous Agents}
\label{sec_DU_eq_price}

Let us first introduce a first trivial result related to trader liquidity.
Recall that $x_t^j$ is the cash endowment of trader $j$ at time $t$. 
All the proofs are reported in Appendix \ref{sec_app_proof}.

\begin{lem}\label{lemma_cash_endowment}
There exists a finite amount of initial cash endowment, denoted as $x_0$, such that each trader can submit at least one buy order at the bid price $q_t=(1-\alpha)u_t+\alpha \overline{p}_{t-1}$ during each trading period $t$ without going bankrupt, i.e., for each agent $j$, $x_t^j\geq0$ for all $t$.
\end{lem}

Thus, in the following we assume that:
 \begin{as}\label{as_cash}
    %At each time $t$, $x_t^j=x_0$ and $y_t^j=1$, i.e.,
    Traders' initial endowments are equal to  $x_0^j=x_0=\kappa FV_0 T$ and $y_0^j=T$ for each agent $j$, i.e., 
     traders have enough endowment to at least\footnote{These values do not represent the
minimum amount of endowment to ensure that traders can submit one buy or sell quote for each period.} submit one buy or sell quote for each trading time period $t$.
 \end{as}

% We first introduce the first assumption related to trader liquidity.
% Recall that $x_t^j$ is the cash endowment of trader $j$ at time $t$.
% The first assumption is:
%  \begin{as}\label{as_cash}
%     At each time $t$, $x_t^j=\infty$, i.e.,
%      traders have enough cash.
%  \end{as}

At first glance, the above assumption seems to limit the insights one can gain from the subsequent analysis. However, this assumption may be valid in a laboratory framework, where an experiment may be designed to guarantee that each participant can actively participate in the market.
For instance, in the SSW experiment design, traders were provided with an endowment of cash and stock quantity equivalent to  about\footnote{Precisely, in one of the SSW experiment designs, three classes of traders were considered with different endowments of cash and share quantity, which on average they correspond to an endowment of $\$13.05$, see, e.g., Table 1 of \cite{duffy2006asset}.} $\$13$, which corresponds to have an initial endowment of $\$ 9.40$ and one stock.
With this kind of endowment, we may expect that in the laboratory, no one of the agents will become bankrupt on average and can actively participate in the market, posting bids and asks quotes. This is what is also observed in the simulation analysis of \cite{cordoni_2021_simulation}, where traders were equipped with an initial inventory of $\$10$ and two stocks and posted at least six outstanding orders for each trading period.

% On the other hand, cash constrains might be negligible under some setting as in optimal execution problems in high-frequency trading, e.g., \cite{almgren2001optimal}.

% Asmp. 1, per non porre limiti di cash constrains (horizon potrebbe essere quello di HFT, in cui i constarain possono essere azzerati e trascurabili) questo può essere evitato in laboratoruo se calibrio bene l'esperimento. L'uso diq eusta ipotesi fa capire come sebbene il modello sia unrealistic il fatto che si osservino bolle senza ipotesi di vincolo di inventario fa capire come
% le bolle si generarno indipendente dai constraints di bilancio ma sono invece dovute ai vari meccanismi di trading che interlacciandosi in orizzonti intraday generaron dinamiche di bolle across days

\begin{de}[Equilibrium]
The state of the market where the supply matches the demand across the entire market is called \textit{equilibrium}.
This equilibrium is achieved when the bid and ask prices are equal and where the quantity supplied equals the quantity demanded for each item being traded.
At this equilibrium market clearance is guaranteed, i.e., there are no excess supplies or excess demands in the market. 
\end{de}

\begin{os}\label{avg_eq_price_os}

In parallel to \cite{ikica2023competitive},
we are interested in studying the average equilibrium price dynamics, which can be recovered as the average of multiple laboratory sessions, i.e., by averaging the price dynamics resulting from different sessions.
This implies that we are interested in the average behaviour of agents.  Therefore, we may consider that a trader can be a buyer or a seller for a specific trading round and submit only one quote, representing the average quote\footnote{See Assumption \ref{as_quotes}.} without loss of generality. This design is similar to batch trading markets architecture and the average equilibrium price will be determined by simply equating
the prevailing bid and ask price, where
these prices will be obtained by equating the
aggregate supply and demand for bid and ask sizes, respectively.
\end{os}

Thus,
from Asm.1, for a seller
$a_t=q_t$ and for a buyer\footnote{Recall from Lemma \ref{lemma_cash_endowment} that if the initial cash endowment is $x_0$ then buyers can submit for each trading period $t$ at least one buy order at bid price $q_t$.} $b_t=q_t$ for all $t.$
Hence, the prevailing bid price\footnote{Each quote is for one asset share.} at time $t$, $p_t^b$, is the solution of
  $$
   \sum_{j=1}^N \frac{\pi_t b_t^j}{p_t^b}=\sum_{j=1}^N \frac{\pi_t q_t}{p_t^b}=A_t
  $$
where $A_t$ is the supply provided by the market equal to $(1-\pi_t) N$.

 In the same way, the prevailing ask price at time $t$, $p_t^a$,
is the solution of
  $$
   \sum_{j=1}^N \frac{(1-\pi_t) a_t^j}{p_t^b}=\sum_{j=1}^N \frac{(1-\pi_t) q_t}{p_t^a}=B_t
  $$
where $B_t$ is the demand provided by the market, which is equal to $\pi_t N$.

% Thus, by considering the average mid price $p_t=(p^b_t+p_t^a)/2$, 

% \begin{equation}
% p_t=\frac{\overline{q}_t}{2}\cdot \left(
% \frac{\pi_t}{1-\pi_t}+\frac{1-\pi_t}{\pi_t}
% \right)
% \end{equation}
% where $\overline{q}_t=(1-\alpha) \kappa \frac{FV_t}{2}+\alpha p_{t-1}$.

% \vspace{10pt}
% Domande: Mi è chiaro che il prezzo di equilibrio si ottiene solamente quando $\pi_t=0.5$, dove la dinamica non dipende da $\pi$ appunto.
% Per $\pi_t=0.5$, la dinamica del modello simulato coincide con quella ottenuta da (1).
% Quando $t\to T$, $p_t\to \frac{\kappa FV_T}{2}.$

% La mia domanda è: cosa succede quando $\pi_t\neq0.5$?
% Precisamente quando $\pi_t=0.5-\phi T$, quello che succede nelle simulazioni è che quando $t\to T$ la probabilità di essere seller aumenterà, quindi ad un certo punto il prezzo crolla generando il pattern classico di bolla.
% Questo è quello che osservo dal modello simulato.
% Tuttavia l'equazione (1) fornisce una dinamica esponenziale esplosiva.

% Dove è che sbaglio?
% Forse è la domanda malposta, perché sto andando fuori equilibrio? Nel caso come faccio ad analizzare questa cosa analiticamente? devo imporre una condizione di sbilanciamento fra ask e bid side.

% Ti torna come ho ragionato per trovare la dinamica del bid e ask price, usando domanda=offerta?

On average an agent places an order 
equal to $$\overline{q}_t=(1-\alpha) \kappa \frac{FV_t}{2}+\alpha \overline{p}_{t-1},$$ so that, on average,
\[
p_t^b=\frac{\pi_t}{1-\pi_t} \overline{q}_t; \quad
p_t^a=\frac{1-\pi_t}{ \pi_t} \overline{q}_t.
\]

\begin{de}[Equilibrium price]
The market-clearing price at \textit{equilibrium}, $\overline{p}_t $, is defined as the price for which $p_t^b=p_t^a$, i.e., when the supply $A_t$ clear the demand $B_t$.
\end{de}

\begin{os}
The equilibrium price will be defined as the price such that bid and ask prices are equal. This notion is different from the standard concept of asset price equilibrium. Perhaps, it would be better to replace the adjective equilibrium and use stationary price instead.
However, the notion of stationary price might generate confusion in a price bubble dynamics framework. Therefore, we will continue to use the adjective of equilibrium, bearing in mind the conceptual difference with the standard notion of equilibrium.
\end{os}

Thus, as argued in Remark \ref{avg_eq_price_os}, we may analyse the average behaviour of agents. So, in the following analysis, for all agents, we consider the average bid/ask quote, i.e., we assume: 
\begin{as}\label{as_quotes}
Each agent submits the average quote
$q_t=\overline{q}_t$ for all $t.$
\end{as}

% Asmp. 2, interpretazione alternativa che c'è un tipo di agente Noise trader che trada più di un asset  dove abbiamo il buyer e il seller che a seconda di pi cambiano i volumi
% typo di numeor di trader e che gli agenti hanno un solo asset

Essentially, we are examing the average submitted quotes for every agent, obtained by averaging across different market sessions according to the definition of agent's behaviour of \cite{duffy2006asset}. Specifically, traders can and will post different quotes for every market session, which can subsequently be aggregated to obtain the average quote $\overline{q}_t$.

Another way with which we might formulate the previous assumption,
and reinterpret the model, is that the traders' population can be divided into two representative agents, a buyer and a seller, which trade with the same average quote but with different volume, $(1-\pi_t)N$ for the seller and $\pi_t N$ for the buyer.
However,  contrary to the original work of \cite{duffy2006asset} and also to \cite{baghestanian2015traders}, in this study, we are not interested in trading volume predictions, but instead we focus on the price dynamics. On the other hand,  we derive interesting insights about the order imbalance dynamics employed to describe a theoretical motivation of price bubbles in experiments. This will provide an additional perspective to the analysis carried out by \cite{baghestanian2015traders} and it is outlined in Section \ref{sec_theoretical_motivation}. 

Therefore, we may state our first results regarding the existence of the equilibrium market-clearing price. 
\begin{te}\label{te_1}
Under assumptions \ref{as_cash}, \ref{as_quotes},
the equilibrium market-clearing price $\overline{p}_t $ does not depend on the number of traders and it exists if and only if $\pi_t=\frac{1}{2}$. Moreover, 
 $\overline{p}_t ={q}_t$.

\end{te}

When $\pi_t\neq 0.5$ the market is not in equilibrium and there is an imbalance between demand  and supply equal to $\left(\frac{B_t}{A_t}-1\right)$ which characterizes the average price dynamics:
\[
p_t=q_t+\left(
\frac{\pi_t}{1-\pi_t}-1
\right).
\]

For instance, when $\pi<0.5$, there are more sellers than buyers on average, so that the imbalance between demand and supply $\left(\frac{B_t}{A_t}-1\right)$ is negative. Therefore, the average price dynamics will result below the quote $q_t$, since the excess supply will push the price dynamics down.
Interestingly, the price dynamics does not depend on the number of traders $N$.
% {\color{blue}
We then analyze the theoretical equilibrium price dynamics 
by varying the model parameters, $\kappa$, $\alpha$ and $\varphi$. To better quantify and visualize the misvaluation effect the Relative Deviation (RD) measure of \cite{stockl2010bubble} is employed. RD satisfies all the evaluation criteria presented in \cite{stockl2010bubble}, i.e., it relates fair value and price, it is monotone and invariant, and it is defined as $    RD = \frac{1}{T}\sum_{t=1}^T \frac{p_t - FV_t}{|\overline{FV}|} = \frac{1}{T}\sum_{t=1}^T RD_t$.

In line with \cite{smith1988bubbles}, \cite{duffy2006asset} and \cite{baghestanian2015traders} the number of trading sessions is set to $T=15$. 
All agents are endowed with enough cash and stocks, according to Assumption 1.
We select the dividend support of asset $P_1$, see Section \ref{sec_market_setup}. We consider as reference parameters the ones estimated on the
\cite{smith1988bubbles} experiments from the \cite{duffy2006asset} calibration, i.e., $\kappa=4$ and $\alpha=0.85$ and $\varphi=0.01$. Therefore, we expect the price dynamics to exhibit the typical bubble-shape of market experiments on average. 

Figure \ref{eq_price_varying_kappa}, \ref{eq_price_varying_alpha} and \ref{eq_price_varying_phi} exhibit the related theoretical average price dynamics when we vary one of the parameters, by fixing the other two. The related RD measure is reported among trading periods. We recall that the price is in equilibrium when $\varphi=0$.

  \begin{figure}[!t]
  \centering
  \includegraphics[width=1\linewidth]{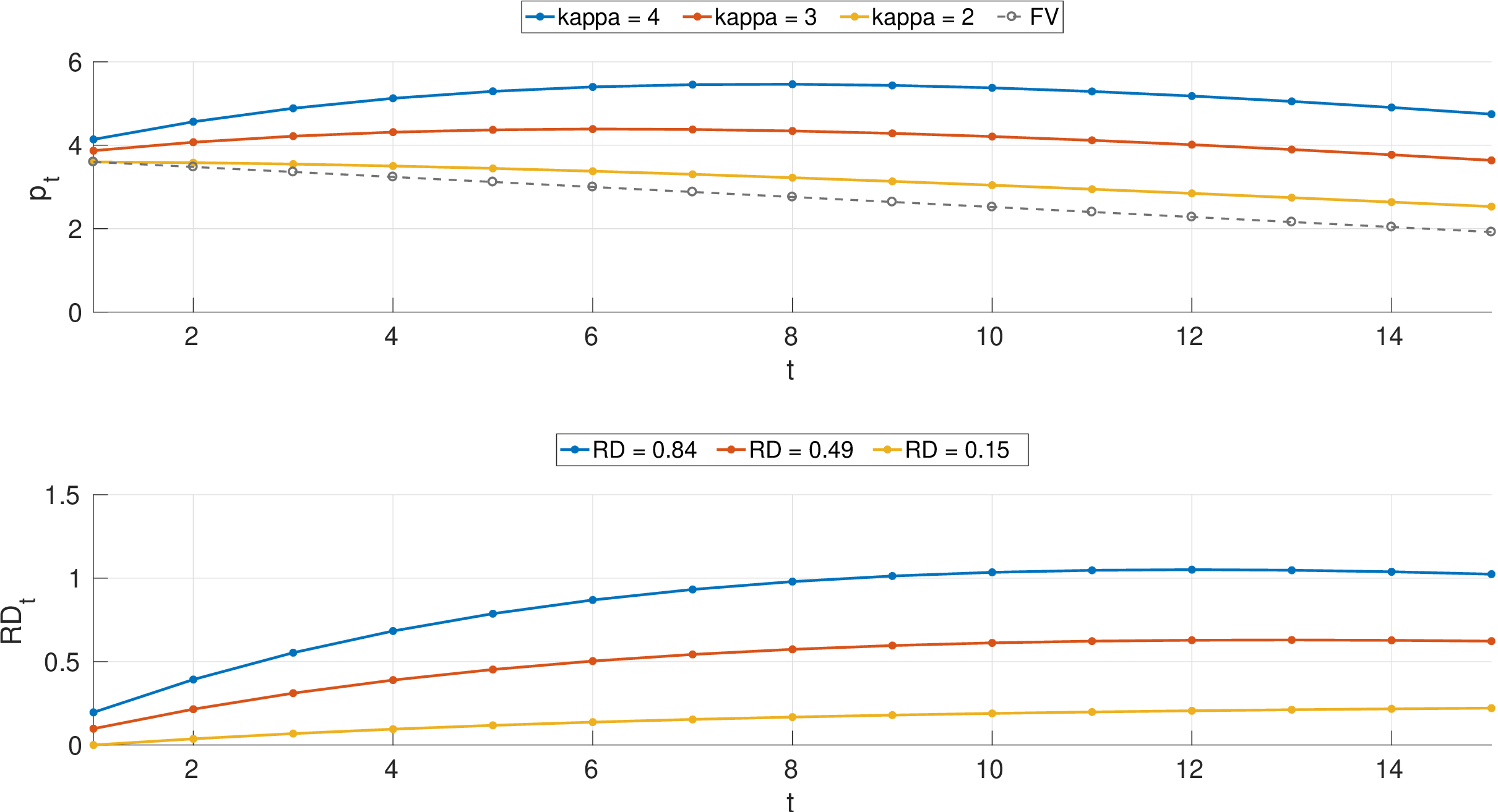}
  \caption{Equilibrium Price ($\varphi=0$) and related RD measure when $\alpha = 0.85$ by varying $\kappa$. The grey line is the fundamental value dynamics. In the legend is reported the average RD measure among the parameter specifications.}
  \label{eq_price_varying_kappa}
\end{figure}

  \begin{figure}[!t]
  \centering
  \includegraphics[width=1\linewidth]{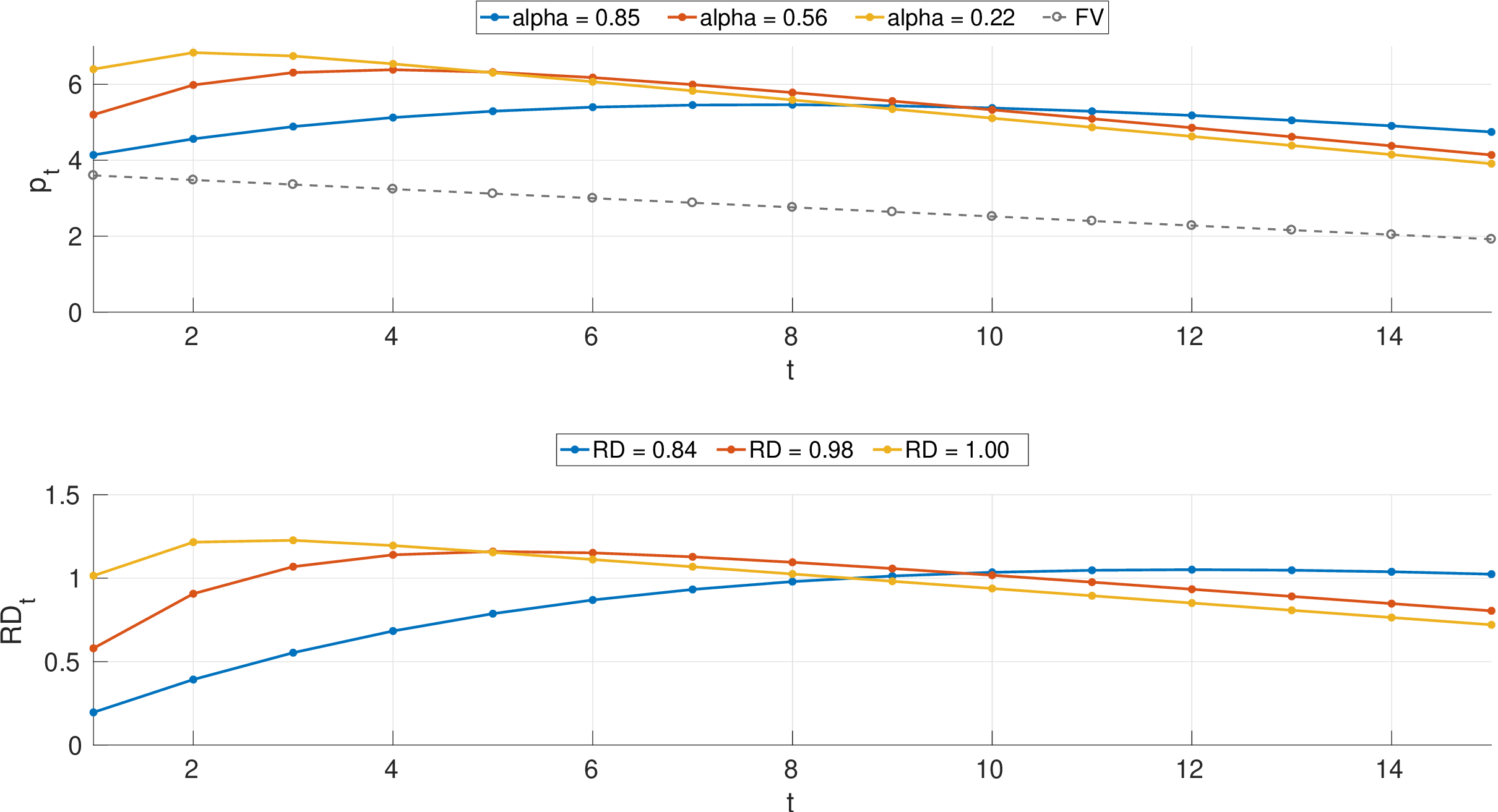}
  \caption{Equilibrium Price ($\varphi=0$)  and related RD measure when $\kappa =4$ by varying $\alpha$. The grey line is the fundamental value dynamics. In the legend is reported the average RD measure among the parameter specifications.}
  \label{eq_price_varying_alpha}
\end{figure}

  \begin{figure}[!t]
  \centering
  \includegraphics[width=1\linewidth]{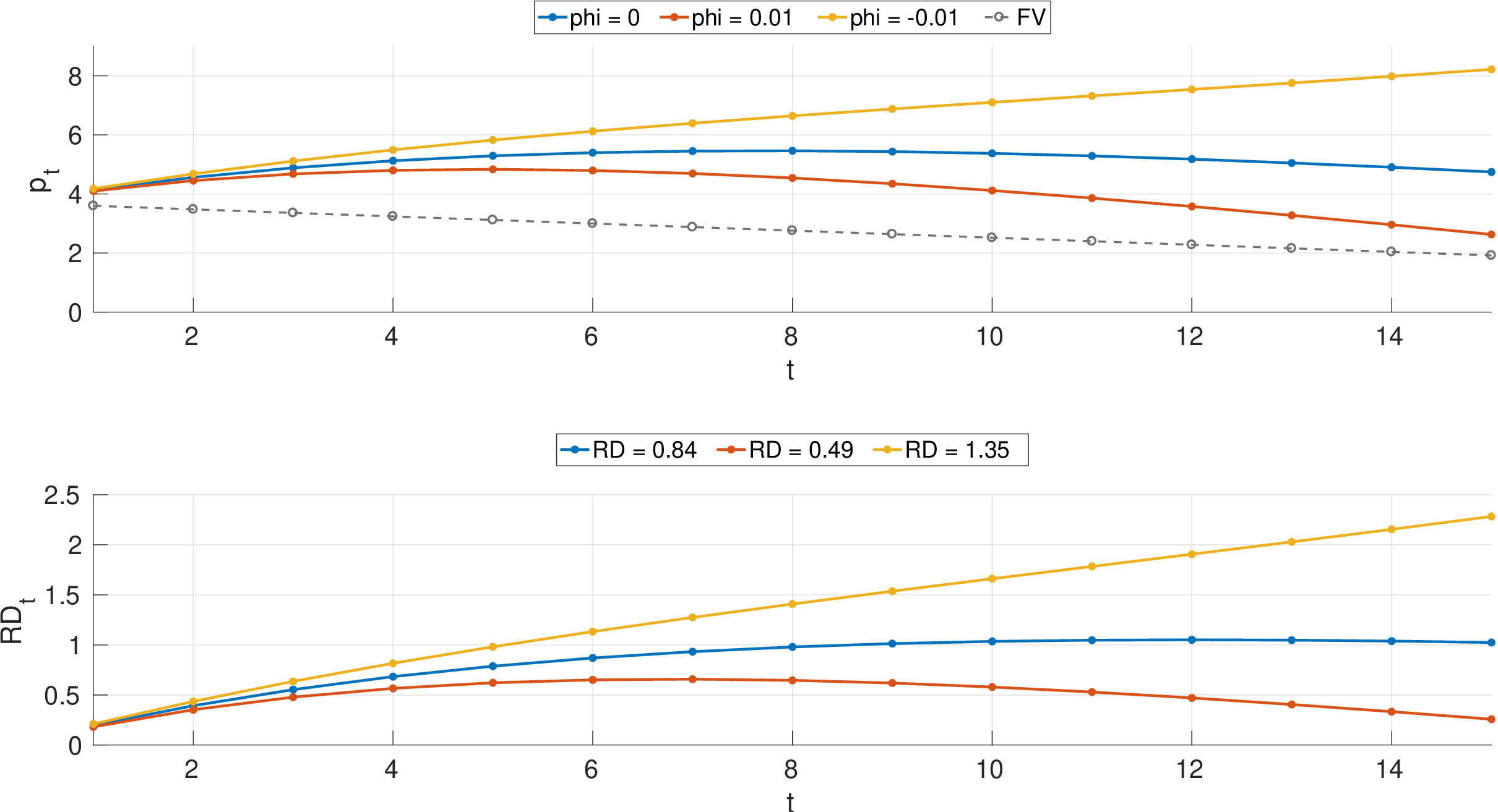}
\caption{Average Price dynamics  and related RD measure when $\kappa = 4$, $\alpha=0.85$ by varying $\varphi$. The grey line is the fundamental value dynamics. In the legend is reported the average RD measure among the parameter specifications.}
  \label{eq_price_varying_phi}
\end{figure}

    \begin{figure}[!t]
  \centering
  \includegraphics[width=1\linewidth]{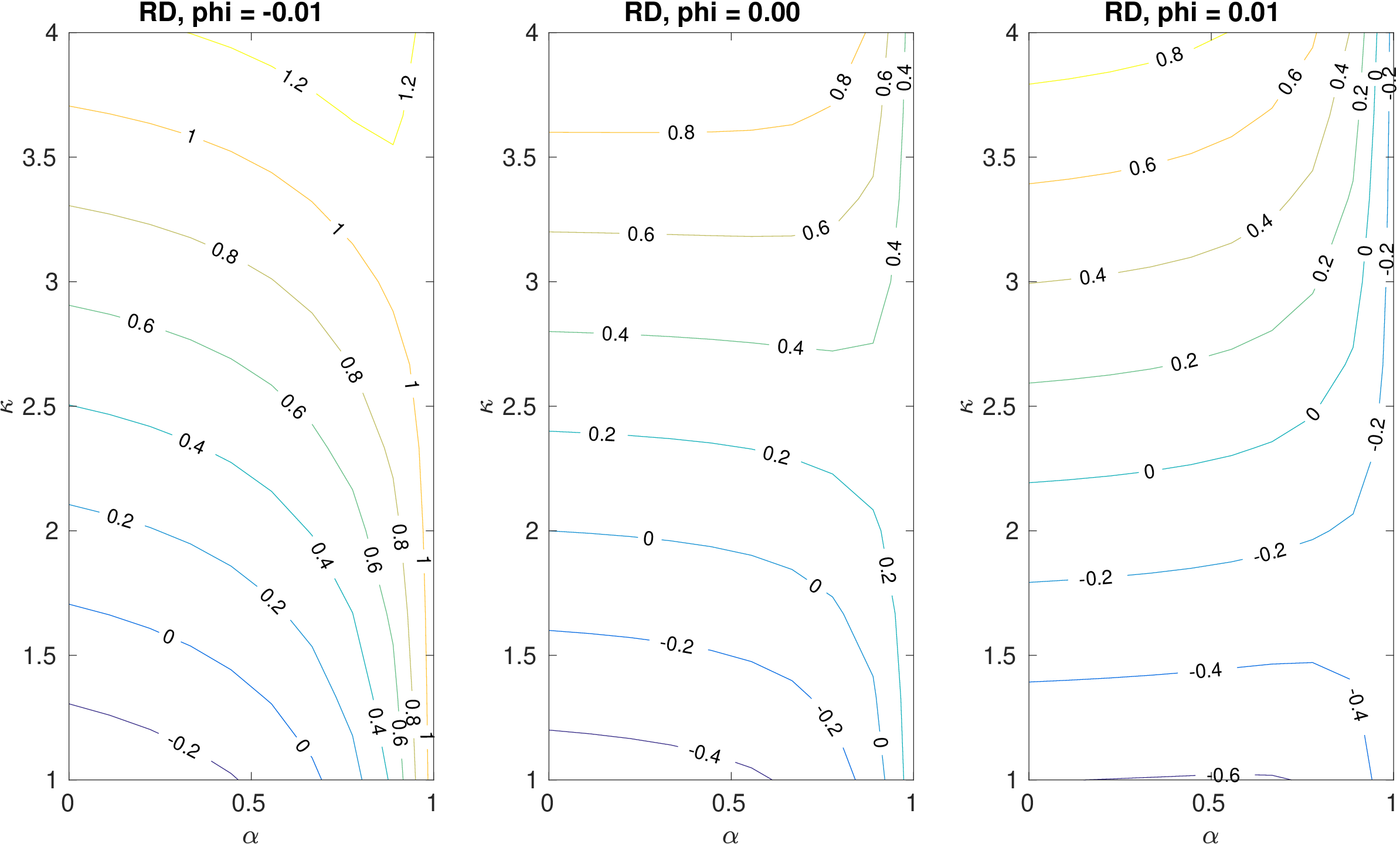}
    \caption{Contour plot of $RD$ by varying $\phi$, $\alpha$ and $\kappa$.}\label{contour_plot_RD_1_dim}
\end{figure}

The overvaluation\footnote{We refer to misvaluation when the price deviates from the related fundamental value. When the prices positively deviate from the fundamental, we say that the asset is overvalued.} measured (at equilibrium) by RD raises when the uncertainty of traders about the fundamental value increases, i.e, when $\kappa$ increases. 
On the other hand, when traders are more anchoring to past prices, i.e., $\alpha$ is close to 1, overvaluation tends to decrease on average, even if it raises in the last trading rounds, see Figures \ref{eq_price_varying_kappa} and \ref{eq_price_varying_alpha}.
% }
% In Figure \ref{fig:eq_price} we compare the theoretical equilibrium price dynamics with those obtained by simulating the DU model.
% In line with \cite{smith1988bubbles},\cite{duffy2006asset} and \cite{baghestanian2015traders} the number of trading session is set to $T=15$. 
% We set $S=6$ with $N=33$ agents for the simulation of DU model. All agents are endowed with enough cash and stocks, according to Assumption 1.
% We select the dividend support of asset $P_1$, see Section \ref{sec_market_setup}. The other parameters are set to $\kappa=4$ and $\alpha=0.8499$. 
% These parameters are estimated on the
% \cite{smith1988bubbles} experiments from \cite{duffy2006asset}. Therefore, we expect the price dynamics to exhibit the typical bubble-shape of market experiments on average. 

% We choose two different values for $\varphi$;
% in the first case, according to the weak foresight assumption, we set $\varphi=0.246/T$, i.e.,
% the market is not at equilibrium since $\pi\neq 1/2$,
% and in the second case, we set $\varphi=0$, i.e., when the market is at equilibrium $\pi=1/2$.
% The dynamics shapes are very similar between the simulated price and the theoretical ones,
% even if there is a positive shift when we simulate the DU model.

We observe that when the market is at equilibrium, the price exhibits a hump-shaped dynamics in accordance with \cite{duffy2006asset}, see Figure \ref{eq_price_varying_phi}. Even if,
 at the middle of the market session, the price reaches a maximum,
no crash is observed at the end, making the asset always overvalued.
Interestingly,  at the end of the market session, the average equilibrium price does not align with the fundamental value.
This general overvaluation can be attributed to the noise traders' risk described by \cite{de1990noise}, where
the noise traders create this difference between price and fundamental value to earn positive returns. 
On the other hand, when $\varphi>0$, we observed a sharp decline of the price, which is aligned with the fundamental value at the end of the market session.

% {\color{blue}
Figure \ref{contour_plot_RD_1_dim} shows the contour plot of the average RD surface among the parameters. Interestingly, by looking the curve levels in the parameter space,
we observe that the weak-foresight assumption ($\varphi>0$) decreases the general overvaluation, see also Figure \ref{eq_price_varying_phi}. 
% }
We remark that regardless of the parameter $\varphi$
the bubble starts close to the fundamental value of the asset, precisely on average the average price of the first trading period will be equal to $FV_1[ (1-\alpha) k/2+\alpha]$. 
Despite the upward trend observed by SSW and \cite{duffy2006asset} in the first trading periods, which is generated by traders' inexperience, is not displayed, the price dynamics clearly exhibit the typical bubble shape. This is consistent with the analysis of \cite{baghestanian2015traders} and  with the results of \begin{NoHyper}\cite{cordoni_2021_simulation}\end{NoHyper}, where, even though the same noise traders of the DU model\footnote{Precisely, they simulate the DU model where $\overline{p}_0=0$ as in \cite{duffy2006asset}.} are employed the price bubbles is aligned with the fundamental value at the beginning of the trading period. This is the result of 
\cite{cordoni_2021_simulation}  market-makers agents employed in their model to provide liquidity to the market, so that they set the average book mid-price of the first trading period to the fundamental value.
Indeed, due to competitiveness, market-makers are forced to trade at efficient prices to avoid to be kicked out of the market.
% \begin{figure}
%     \centering
%     \includegraphics[width=0.9\textwidth]{sim_price_DU.eps}
%     \caption{Theoretical and simulated price dynamics.
%         Dashed lines represent the theoretical price, while solid ones represent the price obtained from simulations.
%         The grey line is the fundamental value dynamics.
%         In blue, the dynamics when
%      $\varphi>0$, i.e., $\pi\neq 1/2$, and in red the equilibrium dynamics when $\varphi=0$, i.e., $\pi=1/2.$ For the simulated dynamics, we average the prices over 100 simulations.
%     }
%     \label{fig:eq_price}
% \end{figure}

%arrivato qui introdurre i fattori 16/11/2021
\section{Heterogeneous Agent Based model: Factor Investing Strategies}
\label{sec:two_assets}

We consider a two-asset extension of the DU model  where three types of agents are introduced: $J_N$ noise traders, $J_D$ directional and $J_{MN}$ market-neutral traders.
The number of traders is equal to $J_N+J_D+J_{MN}=N$.
Following \cite{cordoni_2021_simulation}
we specify two model specifications
for the two assets $P_{1}$ and $P_{2}$,
so that we have two order
books with the relative parameters, $\kappa_{i},\alpha_{i},\varphi_{i}$, for $i=1,2$.
In this model we design multi-asset trading strategies mimicking factor investing style,
see e.g. \cite{li2019transaction}, which traders can implement. 

The near-zero-intelligence agents of  DU
will be used as prototypes of noise traders.  
We assume that the other traders follow one of the assets, i.e., the asset $P_1$, to read a signal to buy or sell, i.e., $s_1\in \{-1,1\}$
where $1$ ($-1$) means that $s_{1}$ is a buy (sell) signal for asset $1$.
As for the noise traders, the probability of reading a buy or sell signal is modelled by $\pi_{1}$,
i.e., the probability of being a buyer or a seller for asset $1$.
The heterogeneity is introduced by considering
a percentage of agents which will follow one of the two market factors:
the \emph{directional} market factor $v_{D}=[1,\ 1]^{T}$ and the
\emph{market-neutral} market factor $v_{M}=[1,\ -1]^{T}$.

 Therefore, a directional (market-neutral) trader places orders on both assets
following the directional (market-neutral) market factor.
Thus, an agent reads the market signal from asset one,
$s_{1}$, to assign the position of buy/sell on $P_1$, while
the position on asset $P_2$ depends on the market factor:
if the trader
is a directional (market-neutral) will place the same (opposite) order side on the other assets,
i.e., the position on both assets are described by the product $s_{i}\cdot v_{D}$ ($s_{i}\cdot v_{M}$).
In other words, a
directional (market-neutral) trader places orders
in asset 2 with the same (opposite) sign position of asset 1.

The quote sizes are the same for all agents, and they are equal to $q_{t,1}$ and $q_{t,2}$ for the two assets, respectively and they might have two distinct parameter specifications. We assume Assumption \ref{as_quotes} for both assets, i.e.,
$q_{t,1}=\overline{q}_{t,1}$ and 
$q_{t,2}=\overline{q}_{t,2}$ for all traders.
We assume that the probability of being a buyer 
for asset $P_1$ is fixed
for all the traders, at trading time $t$, to $\pi_{t,1}$. 
The trading position on asset $2$ for noise traders is assigned by another random variable (independent from asset $P_1$) with a probability of being a buyer given by  $\pi_{t,2}$.
Since
the directional and market-neutral traders will assign asset position on    $P_2$ following the corresponding factor, there is no need to specify another random variable for their positions on asset $P_2.$
Therefore, we require the following assumption.
\begin{as}\label{as_probab_two_assets}
a) At the trading time $t$,
all the traders decide to buy or sell asset $P_1$ following 
i.i.d. Bernoulli random variables with probability $\pi_{t,1}.$
b) At trading time $t$, the noise traders decide to buy or sell asset $P_2$ according to i.i.d. Bernoulli random variables with probability $\pi_{t,2}$.
\end{as}
Then, under Assumptions \ref{as_cash}-\ref{as_quotes}-\ref{as_probab_two_assets} we recover the equilibrium average price dynamics for the two assets.
 For asset $P_1$ 
    traders behave as for the homogeneous case of Section \ref{sec_DU_eq_price}.
    Indeed, if $p_{t,1}^{b}$ is the prevailing bid price at time $t$,
 then, it solves the equation
  $$
    \sum_{j=1}^{J_N} \pi_{t,1} \cdot q_{t,1} 
    + \sum_{j=1}^{J_D} \pi_{t,1}\cdot q_{t,1}+
     \sum_{j=1}^{J_{MN}} \pi_{t,1} \cdot q_{t,1}=A_{t,1}\cdot p_{t,1}^{b}
  $$
where $A_{t,1}$ is the supply provided by the market for asset 1 which is equal to $(1-\pi_{t,1})N$. However,
$ \sum_{j=1}^{J_N} \pi_{t,1}\cdot q_{t,1} 
    + \sum_{j=1}^{J_D} \pi_{t,1}\cdot q_{t,1}+
     \sum_{j=1}^{J_{MN}} \pi_{t,1}\cdot q_{t,1}=N \pi_{t,1} q_{t,1}$, so that $p_{t,1}^{b}=\frac{\pi_{t,1}}{1-\pi_{t,1}}q_{t,1}$.
 In the same way, the prevailing ask price at time $t$, $p_{t,1}^a$, solves the equation
 $$
    \sum_{j=1}^{J_N} (1-\pi_{t,1})\cdot q_{t,1} 
    + \sum_{j=1}^{J_D} (1-\pi_{t,1})\cdot q_{t,1}+
     \sum_{j=1}^{J_{MN}} (1-\pi_{t,1})\cdot q_{t,1}=B_{t,1} p_{t,1}^{a}
  $$
where $B_{t,1}$ is the demand provided by the market which is equal to $\pi_{t,1} N$. So,
$p_{t,1}^{a}=\frac{1-\pi_{t,1}}{\pi_{t,1}}q_{t,1}$.
%quiiiiiiiiiiiiiiii 17/11/2021
Therefore, for Asset $P_1$ the equilibrium price exists when $\pi_{t,1}=0.5$, and in this case it is equal to $q_{t,1}$. The average price dynamics is characterized by 
\begin{equation}\label{eq_price_dyn1}
    p_{t,1}=q_{t,1}+\left(
\frac{\pi_{t,1}}{1-\pi_{t,1}}-1
\right).
\end{equation}

For asset $P_2$, since $\pi_{t,2}$ is the probability to be a buyer for the noise trader, then the prevailing bid price
at time $t$, $p_{t,2}^{b}$, satisfies the equation 
  $$
    \sum_{j=1}^{J_N} \pi_{t,2}\cdot q_{t,2} 
    + \sum_{j=1}^{J_D} \pi_{t,1}\cdot q_{t,2}+
     \sum_{j=1}^{J_{MN}} (1-\pi_{t,1})\cdot q_{t,2}=A_{t,2} p_{t,2}^{b}
  $$
where $A_{t,2}$ is the supply provided by the market for asset 2 which is equal to $J_N (1-\pi_{t,2})+J_D (1-\pi_{t,1})+J_{MN} \pi_{t,1}$. 
We remark that the directional (market-neutral) traders have
the same (opposite) side position for both assets.
Solving for $p_{t,2}^b$, we obtain
  $$ 
  p_{t,2}^{b}=
  \frac{
    J_N\cdot  \pi_{t,2}+J_D\cdot  \pi_{t,1}+
    J_{MN}\cdot  (1-\pi_{t,1})
      }{ J_N \cdot (1-\pi_{t,2})+J_D\cdot (1-\pi_{t,1})+J_{MN}\cdot \pi_{t,1}}
    \cdot q_{t,2} .
  $$
  In analogous way, 
  the prevailing ask price for asset $2$,
  $p_{t,2}^a$, solve the corresponding  equation
   $$
    \sum_{j=1}^{J_N} (1-\pi_{t,2})\cdot q_{t,2} 
    + \sum_{j=1}^{J_D} (1-\pi_{t,1})\cdot q_{t,2}+
     \sum_{j=1}^{J_{MN}} \pi_{t,1}\cdot q_{t,2}=B_{t,2} p_{t,2}^{a}
  $$
where $B_{t,2}$ is the demand for asset 2 which is equal to $J_N \pi_{t,2}+J_D \pi_{t,1}+J_{MN} (1-\pi_{t,1})$. Thus,
the prevailing ask price is equal to 
    $$ 
  p_{t,2}^{a}=
  \frac{
  J_N \cdot (1-\pi_{t,2})+J_D\cdot (1-\pi_{t,1})+J_{MN}\cdot \pi_{t,1}
    }{ J_N\cdot  \pi_{t,2}+J_D\cdot  \pi_{t,1}+
    J_{MN}\cdot  (1-\pi_{t,1})
     }
    \cdot q_{t,2} .
  $$
  Therefore, if the number of directional and market-neutral traders are equal, there exists the equilibrium price for asset $P_2$.  
 \begin{pr}\label{pr_asset1_two_assets}
 Under Asm. 1, 2, 3 and $J_D=J_{MN}$, then there exists an equilibrium for asset 2 if and only if $\pi_{t,2}=0.5$ for all $t$ and 
 $ \overline{p}_{t,2}=q_{2,t}$.
 \end{pr}

 We observe that when $J_D=J_{MN}$ the 
 equilibrium price for the asset $P_2$ is independent of
the probability $\pi_{t,1}$.
Since $\pi_{t,2}$ 
 represents the probability to be a buyer or seller of a noise trader, we may assume that:
 \begin{as}\label{asm_pi}
 $\pi_{t,2}=0.5$ for all $t.$
 \end{as}

Then, we may drop the assumption of $J_D=J_{MN}$ in Proposition \ref{pr_asset1_two_assets}.

\begin{te}\label{te_asset2}
Under Assumptions \ref{as_cash}, \ref{as_quotes}, \ref{as_probab_two_assets}, \ref{asm_pi} and $J_D \neq J_{MN}$,  then there exists an equilibrium price for asset $2$ if and only if $\pi_{t,1}=0.5$ for all $t$ and moreover 
$\overline{p}_{t,2}=q_{2,t}$
\end{te}
Theorem \ref{te_asset2} implies that if $\pi_{t,1}=\pi_{t,2}=0.5$ there exists an equilibrium for both assets, and it is determined by the respective quotes, 
$q_{t,1}$ and $q_{t,2}$.
Moreover, under Assumption \ref{asm_pi}, the imbalance between demand and supply for asset 2 is equal to 
$$
\frac{
J_N \cdot 0.5+J_D \pi_{t,1}+J_{MN}(1-\pi_{t,1})
}{
J_N \cdot 0.5+J_D (1-\pi_{t,1})+J_{MN}\pi_{t,1}
}-1.
$$
Then, under the previous assumptions, the average price dynamic for each of the two assets is given by, respectively,
\begin{equation}
    \begin{split}
        p_{t,1}&=q_{t,1}+\left(
\frac{\pi_{t,1}}{1-\pi_{t,1}}-1
\right)
        \\
        p_{t,2}&=q_{t,2}+\left(
\frac{
J_N \cdot 0.5+J_D \pi_{t,1}+J_{MN}(1-\pi_{t,1})
}{
J_N \cdot 0.5+J_D (1-\pi_{t,1})+J_{MN}\pi_{t,1}
}-1
\right).
    \end{split}
\end{equation}

%arrivato qui 

% {\color{blue}
In Figure \ref{RD_over_time_2_assets_factor} shows
the RD comparison among the two-assets for different model specifications.
We select the assets dividend support as presented in Section \ref{sec_market_setup}.
We follow the experiment design of \cite{cordoni_2021_simulation}, setting
the other parameters to $\kappa_1=4$, $\kappa_2=2$. We set $\alpha_1=\alpha_2=0.85$ and for the equilibrium dynamics, $\pi_1=\pi_2=0.5$. Then, we consider the case when asset $1$ is no longer in equilibrium\footnote{We recall that $\pi_{1,t}=\max\{0.5-\varphi_1 t,0\}$,
so by selecting $\varphi_1>0$, $\pi_1$ is a decreasing function of time. }, i.e., 
$\varphi_1=0.01~0.25/T>0$, for  $J_D=J_{MN}$, $J_D=45>J_{MN}=5$ and $J_D=5<J_{MN}=45$. In both cases the number of noise traders is fixed to $J_N = 50.$

% In Figure 
% and \ref{RD_over_time_2_assets_factor_surf} shows the and contour plot
%rrivato quiiiiiiiiiiii

          \begin{figure}[!t]
  \centering
  \includegraphics[width=0.9\linewidth]{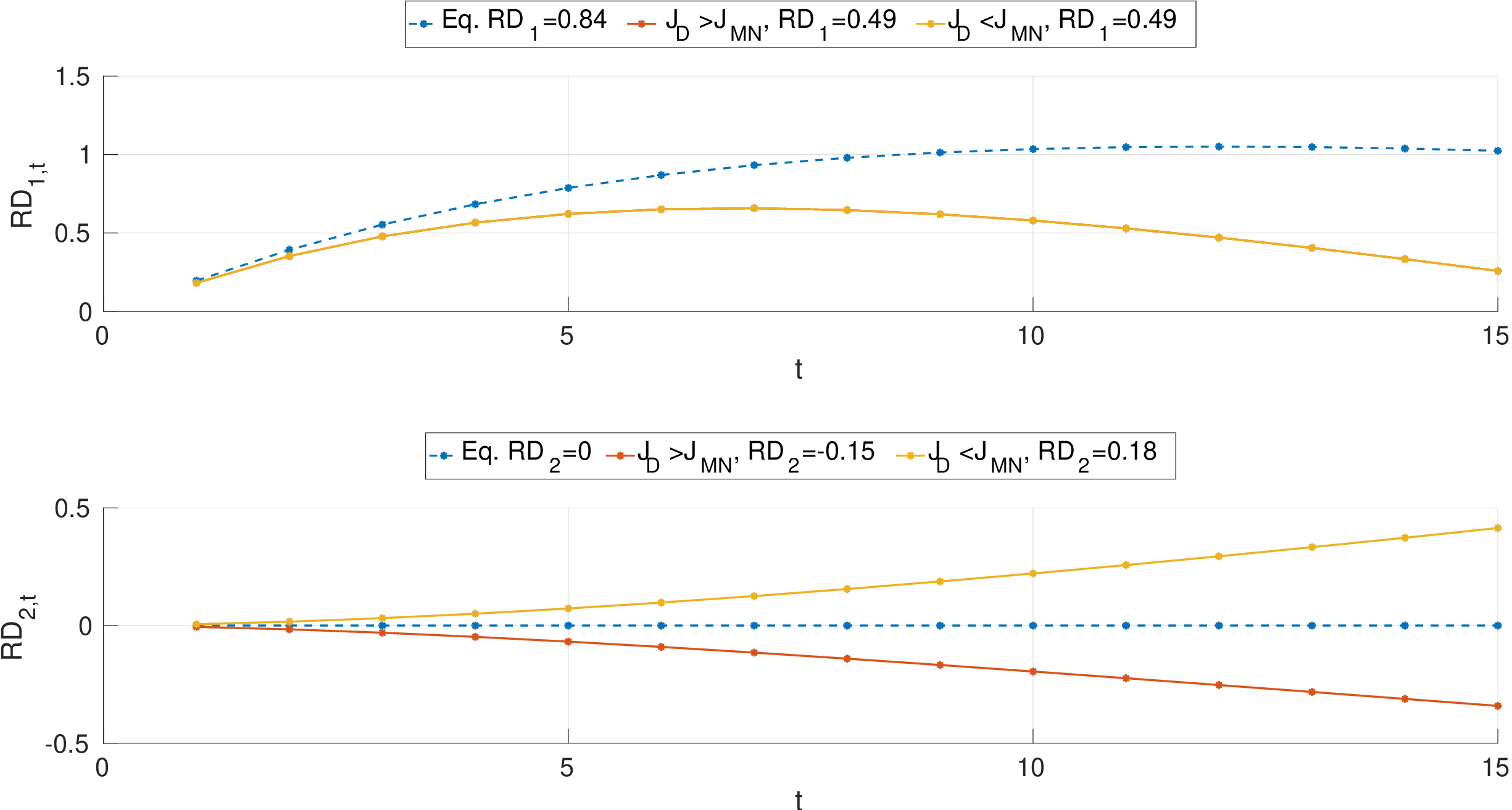}
     \caption{$RD_{i,t}$ for two-asset markets.  Equilibrium prices are obtain for $\varphi_1=0$ (blue dotted lines), and $\varphi_1=0.01$ in other case (solid lines). $\kappa_1 =4$, $\kappa_2 = 2$, $\alpha_1=\alpha_2 = 0.85$ and  $J_D (\%)+J_{MN} (\%) = 50\%$. Red (Orange) lines refer when $J_D>J_{MN}$ ($J_D<J_{MN}$). The average RD, respectively for each case, is reported in the legend bar. In the top exhibit, the $RD_{1,t}$ coincides when  $J_D>J_{MN}$ and $J_D<J_{MN}$.}\label{RD_over_time_2_assets_factor}
\end{figure}

          \begin{figure}[!t]
  \centering
  \includegraphics[width=1\linewidth]{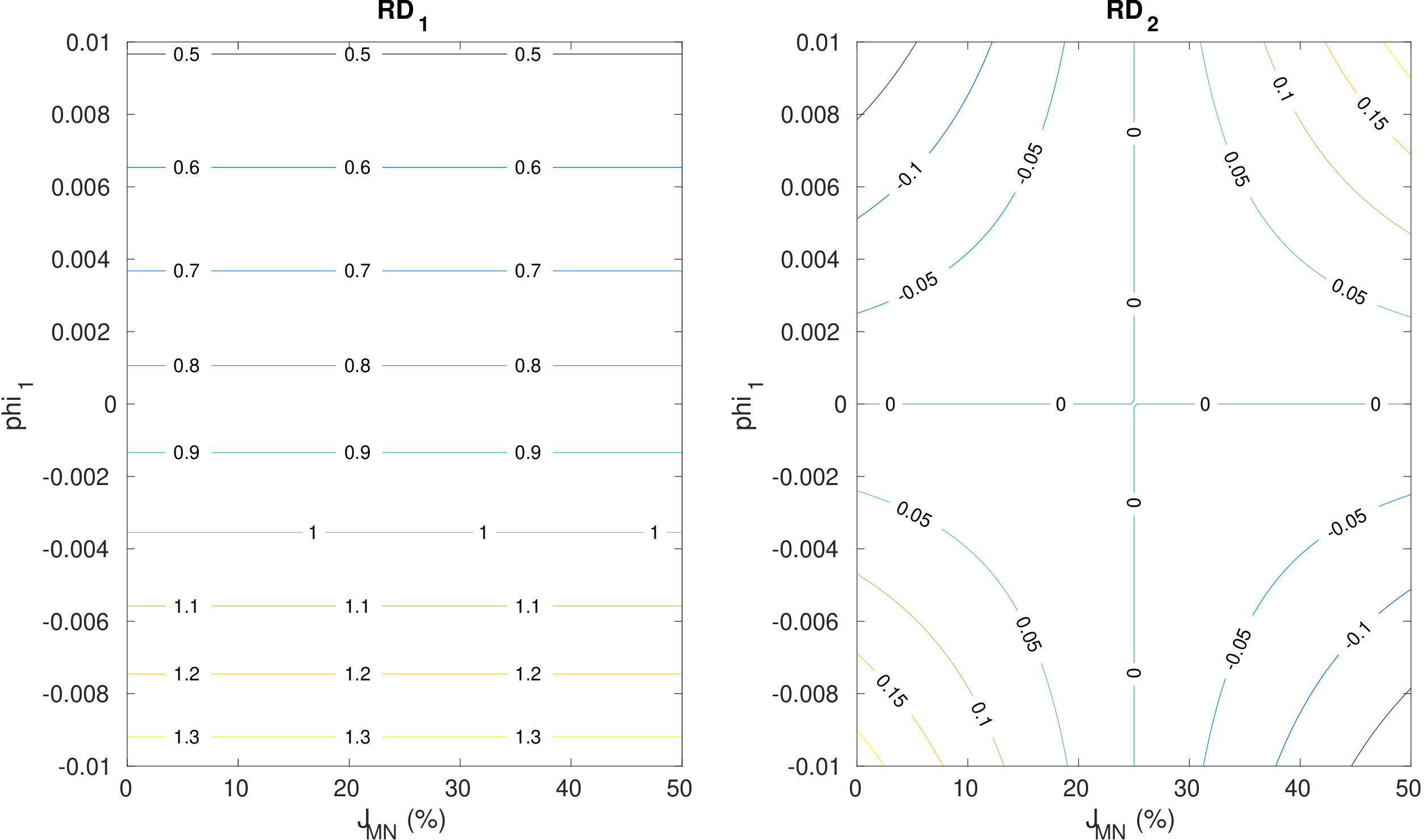}
    \caption{RD for two-asset markets when we vary $\pi_1$ and $J_{MN}$. $\kappa_1 =4$, $\kappa_2 = 2$, $\alpha_1=\alpha_2 = 0.85$ and $J_D (\%)+J_{MN} (\%) = 50\%$.}\label{RD_over_time_2_assets_factor_surf}
\end{figure}

From Figure \ref{RD_over_time_2_assets_factor} and the 
contour plot in Figure \ref{RD_over_time_2_assets_factor_surf}, we observe that when asset 2 is in equilibrium, its price coincides with its fundamental value, i.e., $RD_2=0$, regardless of any parameter setting of $P_1$. Thus, there is no effect of price bubble contagion of asset $P_1$ toward asset $P_2$. Moreover, according to Proposition \ref{pr_asset1_two_assets}, from the right exhibit of Figure \ref{RD_over_time_2_assets_factor_surf} we may observe how the misvaluation is zero when $J_D=J_{MN}$ regardless if the asset $P_1$ is in equilibrium, i.e. $\varphi>0$.
Furthermore, we may also observe how the misvaluation of asset $P_1$ is invariant from the percentual of directional and market-neutral traders.

On the other hand, when $\varphi_1>0$ and $\pi_2=0.5$, we observe that the price bubble of asset $P_1$ affects the dynamics of asset $P_2$ when the proportion between directional and market-neutral traders is varying, generating a misvaluation effect. In particular, when there are more market-neutral agents in the market than directional traders, the bubble of $P_1$ triggers a ``overvaluation" effect also for asset $P_2$, by positively deviating the price from its fundamental value. Viceversa, when $J_D>J_{MN}$ we observe an ``undervaluation" effect for $P_2$.
These findings are consistent with what was observed in the simulation study of \cite{cordoni_2021_simulation}.

For both Figures \ref{RD_over_time_2_assets_factor} and \ref{RD_over_time_2_assets_factor_surf}, we observe that when the market is not in equilibrium, the price-bubble, i.e., asset $P_1$, exhibits a sharp decline at the end of the session aligning with the fundamental value.

\section{Equilibrium Price for Heterogeneous Agent-Based model with investment strategies}
\label{sec_hetero_theory}
% \input{sec_heterogenous_model}
% \section{Heterogenous Agent Based model with investment strategies-equilibrium price}

The parameter $\varphi>0$ plays a crucial role in the DU model in order to get
consistent results with the experimental data. 
It reduces the transaction
volume over time consistent with the experimental data through the weak foresight assumption. 
The DU model relies on this assumption to generate the observed crash patterns of the laboratory market experiments.
However, this artificial hypothesis is nothing else than a pure statistical condition that guarantees a progressively decreasing prices and volume transactions in an exogenous way.
We now drop this assumption by considering the heterogeneous model
of \cite{baghestanian2015traders} and analyzing the corresponding equilibrium price dynamics in a two-asset market employing (endogenous) investment strategies.

We first present the model for a generic asset (without specifying the subscript index), and then we specify how we extend these strategies to the two-asset market. 
We consider three types of agents: $J_N$ noise traders, $J_F$ fundamentalist and $J_{S}$ speculative traders.
The number of traders is set to  $J_N+J_F+J_{S}=N$.
 % The fundamentalist and speculative traders track the fundamental value and past prices to decide their position.
The definitions of fundamentalist and speculator agents are aligned with the work of \cite{baghestanian2015traders}, drawing inspiration from the studies of \cite{cason1992call} and \cite{haruvy2006effect}. 
 Fundamental traders tend to buy when the price is below what they believe to be the intrinsic value and sell when the price is above this value. 
On the other hand, speculators form their price expectations by accounting for the presence of noise traders, similar to Level-1 traders, see \cite{baghestanian2015traders}. Their trading decisions are then guided by the anticipation of capital gains, leading them to buy when they foresee price increases and sell otherwise.
 Their quote sizes are denoted by $q_t^F$ and $q_t^S$, respectively.

Following \cite{haruvy2006effect} and \cite{baghestanian2015traders}, the fundamentalists compute in every trading period $t$ a proxy measure
for the expected market-clearing price in period $t$, denoted by $l_t$. This proxy is
linked to the fundamental value and past trading price,
$l_t=\alpha^F l_{t-1}+(1-\alpha^F) p_{t-1} - \overline{d}$, 
where $\alpha^F \in (0,1)$ and $l_0=FV_1+\overline{d}$.\footnote{
Since $FV_t=FV_{t-1}-\overline{d}$, in the definition of $l_t$ we subtract $\overline{d}$ to control for the decreasing of the fundamental value. In the extreme case of $\alpha^F=1$, for each $t$, $l_t=FV_t$. We study this case in Section \ref{sec_multi_scenario_puzz}.}
If $l_t\leq FV_t$, they decide to submit a buy order, otherwise, they submit a sell order. The quote size, under Assumption \ref{as_cash} and \ref{as_quotes}, 
is on average $$q_t^F =\frac{l_t + FV_t}{2}.$$

The speculative traders decide whether to buy or sell depending on their expectations about clearing prices in period $[t, t+1]$ at the beginning of the trading period $t$.
 The speculator trading behaviour model follows from \cite{haruvy2006effect} and \cite{baghestanian2015traders}. We employ a Level-k modeling framework to determine speculators' expectations, see, e.g., \cite{stahl1995players} and \cite{crawford2007level}. Therefore, speculators behave as Level-1 traders, reacting optimally against a baseline group of Level-0 noise traders. The average equilibrium price process in a market composed entirely of noise traders follows a certain pattern, as shown in Section \ref{sec_DU_eq_price}. However, assuming that speculators know precisely the noise trader parameters $\alpha$ and $\kappa$, which characterize the equilibrium dynamics is too strong.

Thus, speculative traders form expectations of future prices using a linear combination of 
previous trading price and fundamental value:
\[
E_{p_t}=\gamma_1 p_{t-1}+\gamma_2 FV_t, \gamma_1 \in [0,1], \ 
\gamma_2\geq0.
\]
Iterating one period forward we may obtain $E_{p_{t+1}}$ in function of $FV_t, FV_{t+1},p_{t-1},p_{t-2}$.
A large value of $\gamma_1$  ($\gamma_2$) means that the speculators place more weight on the previous period price (fundamental value) when forming expectations about the current period price.
If $E_{p_{t+1}}>E_{p_{t}}$ the speculators will post a bid otherwise they will post an ask.
Interestingly, in the experiments conducted by \cite{smith2014irrational}, aggregate neural activity seems to serve as an indicator (biomarker) for price bubbles, aligning with historical narratives of euphoria and irrational exuberance near price peaks. Consequently, traders may experience cognitive dissonance when reconciling fundamental and social valuations, further complicating decision-making processes. In a way, this aspect is integrated into the trading behaviour model of speculators, as they combine these two valuation types.
Their quotes, under Assumption \ref{as_cash} and \ref{as_quotes}, 
are on average $$q_t^S =\frac{E_{p_{t+1}}+E_{p_{t}}}{2}.$$

Summarizing, the distinction between fundamentalists and speculators lies in their trading behaviour when making buy or sell decisions. Fundamentalists buy assets when their expected market-clearing price, denoted by $ l_t$, is lower than the fundamental value. In contrast, speculators base their decisions on future price expectations, determining whether to buy or sell accordingly.

The quote sizes for noise agents are equal to $q_{t}=(1-\alpha)u_t+\alpha \overline{p}_{t-1}$, where $u_t$ is a random variable with an average equal to $\kappa FV_t/2$, as described in Section \ref{sec_DU_model_start}.
Therefore, let $p_t^b$ and $p_t^a$ the prevailing bid and ask price
\[
\sum_{j=1}^{J_N} \pi_t q_t +\sum_{j=1}^{J_F}
1_{l_t\leq FV_t} q_t^F+ \sum_{j=1}^{J_S} 1_{E_{p_{t+1}}>E_{p_{t}}} q_t^S=A_t p_t^b
\]
\[
\sum_{j=1}^{J_N} (1-\pi_t) q_t +\sum_{j=1}^{J_F}
1_{l_t> FV_t} q_t^F+ \sum_{j=1}^{J_S} 1_{E_{p_{t+1}}\leq E_{p_{t}}} q_t^S=B_t p_t^a
\]

On average,
noise traders act as liquidity providers for  
fundamentalists and speculators
and we may assume that $J_N\geq J_F+J_S$, where the probability to buy an asset, $\pi_t$, for noise traders is equal to $0.5$.
The equilibrium is recovered when $J_F=J_S=0$, and in this case, the model is exactly the DU model, where the average equilibrium price is equal to 
$\overline{p}_t=q_t$,
 see also the discussion in \cite{baghestanian2015traders}.

 \begin{os}[Market Selection Equilibria]\label{remark_evolution}
Noise traders are the main cause of mispricing.
This is not surprising, see, e.g.,  \cite{black1986noise}, and in experimental laboratory markets, 
one approach to mitigate bubbles involves gradually reducing the presence of noise traders by using experienced subjects, \cite{haruvy2007traders}.

Another approach for reducing the impact of noise traders involves implementing co-evolutionary mechanisms within the trader population \cite{dosi1994introduction}, \cite{farmer2002market}, \cite{bottazzi2013intro_special}. Analogous to the fluctuation of species in an ecosystem, traders may shift between roles (noise traders, fundamentalists and speculators) in response to changing market conditions. For instance, noise traders could evolve into other species or become extinct, thereby reducing their impact on prices. 
However, it is unclear how this co-evolution should occur. For example, it is unclear whether fundamentalists or speculators would prevail. 
This aspect is somewhat reflected in the subsequent section, where towards the end of the trading period, speculators behave as fundamentalists, both contributing to the reduction of order imbalance.
This analogy to agent evolution has also been explored in experimental studies, see, e.g., \cite{hommes2009complex}. 

%-----------------------------------
% In standard microstructural models, heterogeneity is endogenously imposed by distinguishing agents between those with information and those without, as in \cite{kyle1985continuous}. 
% In evolutionary finance, the transition process is akin to ecological succession, where the composition of species changes over time due to various factors like competition, resource availability, and environmental changes. Similarly, traders may evolve from noise traders to more informed participants as they gain experience and adapt to market conditions.

In our model, although information is shared among heterogeneous agents, heterogeneity is imposed by specifying how agents make trading decisions, thereby defining their trading species. Co-evolution should occur by modifying traders' behaviour, examining how learning, adaptation, and competition shape the behaviour of market participants over time, in the spirit of \cite{farmer2002market}, \cite{farmer2009economy}. This evolutionary process will determine the emergence of evolutionarily stable investment strategies, \cite{hommes2014behavioral}.

In our study, individual species are fixed without considering co-evolutionary dynamics. A starting point for a further extension of this research could be to follow the approach of \cite{bottazzi2014evolution}.
\cite{bottazzi2014evolution} investigated wealth-driven selection criteria and found local stability of market selection equilibria represented by generalized Kelly rules in a repeated market with short-lived assets.
 Interestingly, their results highlight the coexistence of stable and unstable equilibria leading to persistent heterogeneity, path dependency, and mispricing.
 Specifically, they show that
 traders with perfect knowledge of underlying dividend processes are not guaranteed to survive unless they possess perfect knowledge and exploit it at best employing a price-dependent generalization of the Kelly rule, termed the S-rule (refer to Section 5 of \cite{bottazzi2014evolution} for complete details).
 In such cases, the selection equilibria where these traders survive are the unique stable equilibria, aligning asset prices with their fundamental values over the long run and thereby reducing RD to zero. We postpone the extension of this work with evolutionary aspects of trading behaviour to further research.
 
 \end{os}

\subsection{A theoretical motivation of Price Bubble in Experiments}\label{sec_theoretical_motivation}
%arrivato qui final lecture
We now combine the previous equilibrium results and trader investment strategies to motivate the typical hump-shaped price during market experiments.
We explain the behaviour discussed in \cite{baghestanian2015traders} from a theoretical point of view, analyzing the trading events which lead to price bubbles shape.

When $J_F+J_S>0$, we compute the relative order imbalance that characterizes the price dynamics.
However, we have to consider four possible events and compute for each event the corresponding imbalance. The events corresponds to when the fundamentalists and speculators are buyer and/or seller.
In event
$E_1=\{l_t \leq  FV_t, E_{p_{t+1}}\leq E_{p_{t}}\}$,
 fundamentalists will buy and speculators will sell, in
$E_2=\{l_t \leq  FV_t, E_{p_{t+1}}>E_{p_{t}}\}$,
both fundamentalists and speculators will buy, while 
$E_3=\{l_t >  FV_t, E_{p_{t+1}}\leq E_{p_{t}}\}$ both fundamentalists and speculators will sell,
and finally in
$E_4=\{l_t >  FV_t, E_{p_{t+1}}>E_{p_{t}}\}$
 fundamentalists will sell and speculators will buy.
Then, respectively for each event, we may derive the 
demand and supply imbalance,
% \[
% \begin{split}
% E1&: \begin{cases}
% J_N \pi_t q_t +J_F
%  q_t^F   &=A_t p_t^b\\
%  J_N (1-\pi_t) q_t  + J_S  q_t^S&=B_t p_t^a
%  \end{cases}, 
% \text{where $A_t=J_N(1-\pi_t)+J_S$, $B_t=J_N \pi_t+J_F$.}
% \\
% E2&: \begin{cases}
% J_N \pi_t q_t +J_F
%  q_t^F    +J_S  q_t^S&=A_t p_t^b\\
% J_N (1-\pi_t) q_t  &=B_t p_t^a
%  \end{cases}, 
% \text{where $A_t=J_N(1-\pi_t) $, $B_t=J_N \pi_t+J_F+J_S$.}
% \\
% E3&: \begin{cases}
% J_N \pi_t q_t    &=A_t p_t^b\\
%  J_N (1-\pi_t) q_t+J_F
%  q_t^F  + J_S  q_t^S&=B_t p_t^a
%  \end{cases}, 
% \text{where $A_t=J_N(1-\pi_t)+J_F+J_S$, $B_t=J_N \pi_t$.}
% \\
% E4&: \begin{cases}
% J_N \pi_t q_t +J_S  q_t^S&=A_t p_t^b\\
% J_N (1-\pi_t) q_t+ J_F
%  q_t^F     &=B_t p_t^a
%  \end{cases}, 
% \text{where $A_t=J_N(1-\pi_t)+J_F $, $B_t=J_N \pi_t+J_S$.}
% \end{split}
% \]
\[
\begin{split}
E1&: \begin{cases}
A_t&=J_N(1-\pi_t)+J_S\\
 B_t&=J_N \pi_t+J_F
 \end{cases}, 
 \text{ then } 
\frac{B_t}{A_t}-1=\frac{0.5 J_N+J_F}{0.5 J_N+J_S}-1.
\\
E2&: \begin{cases}
A_t&=J_N(1-\pi_t)\\
 B_t&=J_N \pi_t+J_F+J_S
 \end{cases}, 
 \text{ then } 
\frac{B_t}{A_t}-1=\frac{0.5 J_N+J_F+J_S}{0.5 J_N}-1=
2\frac{ J_F+J_S}{  J_N}
\\
E3&: \begin{cases}
A_t&=J_N(1-\pi_t)+J_F+J_S\\
 B_t&=J_N \pi_t
 \end{cases}, 
 \text{ then } 
\frac{B_t}{A_t}-1=\frac{0.5 J_N}{0.5 J_N+J_F+J_S}-1
\\
E4&: \begin{cases}
A_t&=J_N(1-\pi_t)+J_F\\
 B_t&=J_N \pi_t+J_S
 \end{cases}, 
 \text{ then } 
\frac{B_t}{A_t}-1=\frac{0.5 J_N+J_S}{0.5 J_N+J_F}-1.
\end{split}
\]

The price dynamics is path-dependent and 
characterized by the parameters of the fundamentalist and 
speculative traders. Therefore, at the end of each time step, we have to compute the position and beliefs of fundamentalists and speculators to decide which 
one of the events we are and compute the corresponding demand and supply imbalance.
Trivially, the imbalance pushes up and down prices depending on 
the number of fundamentalists and speculators.
We characterize
the price dynamics by recovering the bid and ask prices due to the average quotes among noisy, fundamentalists and speculators quotes depending on the event realization.
We recall that the market is in equilibrium when $p_t^a=p_t^b$.
The price dynamics of (prevailing) ask and bid can be expressed among the events in the following way,

\[
\begin{split}
E1&: \begin{cases}
p_t^a&=\frac{0.5 J_N q_t+J_S q_t^S}{
0.5 J_N +J_F}\\
p_t^b&=\frac{0.5 J_N q_t+J_F q_t^F}{
0.5 J_N +J_S}
\end{cases},
\quad E2 : \begin{cases}
p_t^a&=\frac{0.5 J_N q_t }{
0.5 J_N +J_F+J_S}\\
p_t^b&=\frac{0.5 J_N q_t+J_F q_t^F+J_S q_t^S}{
0.5 J_N }
 \end{cases}, 
\\
E3&: \begin{cases}
p_t^a&=\frac{0.5 J_N q_t+J_F q_t^F+J_S q_t^S }{
0.5 J_N  }\\
p_t^b&=\frac{0.5 J_N q_t}{
0.5 J_N +J_F+J_S}
 \end{cases}, 
\quad
E4 : \begin{cases}
p_t^a&=\frac{0.5 J_N q_t+J_F q_t^F}{
0.5 J_N +J_S}\\
p_t^b&=\frac{0.5 J_N q_t+J_S q_t^S}{
0.5 J_N +J_F}
 \end{cases}, 
\end{split}
\]

% \[
% \begin{split}
% p_t&=q_t+ 
% \left(
% \frac{0.5 J_N+J_F}{0.5 J_N+J_S}-1
% \right)1_{E_1} (t)+
% \left(
% 2\frac{J_F+J_S}{J_N}
% \right)1_{E_2} (t)+ 
% \left(
% \frac{0.5 J_N}{0.5 J_N+J_F+J_S}-1
% \right)1_{E_3} (t)\\ &+ 
% \left(
% \frac{0.5 J_N+J_S}{0.5 J_N+J_F}-1
% \right)1_{E_4} (t)
% \end{split}
% \]

% We observe that the price is still in equilibrium when $J_F=J_S$.

% \begin{lem}
% If $J_F=J_S$ then the price is in equilibrium in events $E_1$ and $E_4$.
% \end{lem}

% \begin{proof}
% The supply and demand imbalance vanishes in $E_1$ and $E_4$ when $J_F=J_S$.
% \end{proof}
% Furthermore, in the limit case 
% when $J_F+J_S=J_N$ and 
% $J_F=J_S$, the price dynamics simplified
% as follows 
% \[
% p_t=q_t+2\cdot 1_{E_2}(t)-\frac{2}{3}\cdot1_{E_4}(t).
% \]

% {\color{blue}
When $J_N\to \infty$ 
( and $J_F+J_S$ is bounded)
the average price dynamics will converge on the equilibrium price $\overline{p}_t=q_t$, dotted blue lines of Figure \ref{fig:eq_price_puzzello_fund}, top panels, characterized by noise traders' activity, otherwise, the average price dynamics will be determined by the (mid-)price formed by the interaction of all traders, noise, fundamentalists and speculators, red lines of Figure \ref{fig:eq_price_puzzello_fund}, top panels.
We may formalize the previous statement as follow.
Let $\overline{p}_t^{Het}$ and $\overline{p}_t^{Hom}$ denote the average market-clearing price of the heterogeneous and homogeneous model, respectively.

\begin{pr}\label{pr_average_price_NOISE}
Under Assumptions 1, 2 and previous model specifications, if $J_F+J_S$ is bounded, when $J_N\to \infty$ the market is in equilibrium and the average market-clearing price of the heterogeneous model will converge on the equilibrium market-clearing price of the homogeneous model, i.e., $ \overline{p}_t^{Het} \to \overline{p}_t^{Hom}$.
\end{pr}

% {\color{blue}
To illustrate the convergence of the mid prices $\overline{p}_t^{Het} \to \overline{p}_t^{Hom}$, according to Proposition \ref{pr_average_price_NOISE}, the proportion of noise traders $J_N$ should be a large number compared to $J_F + J_S$. Therefore, for illustration purposes, we select $J_N = 90\%$ of the total population.
Figure \ref{fig:eq_price_puzzello_fund}, \ref{fig:eq_price_puzzello_eq} and \ref{fig:eq_price_puzzello_spec} display the average mid-price and equilibrium price dynamics by computing the ask and bid along with the order and cumulative imbalance, for different 
proportions of $J_F$ and $J_S$.
 The noise trader quotes are updated using the previous period mid-price, i.e., $q_t=(1-\alpha)\overline{u}_t+\alpha (p^a_{t-1}+p^b_{t-1})/2 $.
The average equilibrium price $\overline{p}_{t}$ is computed recursively, 
$\overline{p}_{t}=(1-\alpha)\overline{u}_t+\alpha \overline{p}_{t-1}$, where $\overline{p}_{0}=FV_1.$ 
 We set
 $\kappa=4$, $\alpha=0.85$, $\alpha^F=0.25$, $\gamma_1=0.10$ and $\gamma_2=4$.\footnote{
  The parameters are consistent with the estimates provided by
 \cite{duffy2006asset} (for the noise traders parameter) and \cite{baghestanian2015traders} (for the fundamentalists and speculators).}
Given the total population $J_{total}$, we selected the noise trader population to be $J_N=90\% J_{total}$, and we vary the fraction of fundamentalist and speculator to $J_F=6\% J_{total}$ and $J_S=4\% J_{total}$, Figure \ref{fig:eq_price_puzzello_fund}, $J_F=5\% J_{total}$ and $J_S=5\% J_{total}$,  Figure \ref{fig:eq_price_puzzello_eq}, and $J_F=4\% J_{total}$ and $J_S=6\% J_{total}$, Figure \ref{fig:eq_price_puzzello_spec}.
We select the asset dividend support of $P_1$ as in Section \ref{sec_market_setup}.
% We first consider the limit case when fundamentalists and speculator agents form expectations for the next
%  market-clearing price 
%  without considering the previous trading period price, i.e., $\alpha^F=1$ and $\gamma_1=0$.
%  We set
%  $\kappa=4$, $\alpha=0.85$, $J_N=50$, 
%  $J_S=J_F=25$, $\alpha^F=1$, $\gamma_1=0$ and $\gamma_2=4$.
%  The parameters are consistent with the estimates provided by
%  \cite{duffy2006asset} (for the noise traders parameter) and \cite{baghestanian2015traders} (for $\gamma_2$).
% We select the asset dividend support of $P_1$ as in Section \ref{sec_market_setup}.

The price dynamics exhibits the typical hump price-bubble shape.
 As observed by  \cite{baghestanian2015traders}, 
 assuming that $J_F>J_S$,
 when $\alpha^F<1$ and $\gamma_1>0$
 the initial phase is characterized by an accumulation of shares.
 In \cite{baghestanian2015traders} fundamentalists buy from 
 speculators and noise traders, i.e., we are in event E1. From the bid/ask imbalance, $\frac{0.5J_N+J_F}{0.5 J_N+J_S}-1$, we may observe that the price is pushed forward. However, as discussed in \cite{smith1988bubbles}, \cite{duffy2006asset} agents start trading the stock at a low value
compared to the fundamental due to inexperience. Thus, by assuming that $\overline{p}_0=FV_1$, we are implicitly assuming that traders are in some way experienced enough to correctly compute the fundamental value at time $t=1$.\footnote{In all experiments setting the information about fundamental value is available to all players.} 
 Therefore, in our setting the first event which is realized is event E2,
  where fundamentalists together with speculators decide to buy, see, e.g, Figure \ref{fig:eq_price_puzzello_fund}, and traders generate an upward trend with a subsequent soaring of the price dynamics.
 This triggers the \emph{boom} phase, where the price is pushed away from the fundamental value with an imbalance equal to $2\frac{J_F+J_S}{J_N}$. Thus, since the price will be far away from the fundamental value,
 the fundamentalists decide to sell to the speculators and noise traders, event E4.
 This event is realized in the middle of the trading session until the price reaches its peak. Then, the price starts its decline driven by the imbalance 
 $\frac{0.5J_N+J_S}{0.5J_N+J_F}-1$.
Subsequently,  
 also, speculators start to sell together with fundamentalist, i.e.,  we are in event $E3$, the \emph{burst} phase, with a consequent liquidity drop fulfilled by noise traders, which causes the price-bubble crash,  supported by the imbalance of 
 $\frac{0.5J_N}{0.5J_N+J_F+J_S}-1$.  
 Speculators start to sell since it decreases the traders' subjectively perceived probability of being able to sell \citep{smith1988bubbles,duffy2006asset, baghestanian2015traders}.
The burst phase starts when the cumulative imbalance becomes negative, see, e.g.,  bottom panel of Figure \ref{fig:eq_price_puzzello_fund}.
 
 We observe that during events E2 and E4, the spread is closed, i.e., $p_t^a\leq p_t^b$ generating the equilibrium price, until event E3 starts where the spread will be open. In this phase we may consider the equilibrium mid-price dynamics, which is determined by noise traders' orders which are executed inside the spread between $p_t^b$ and $p_t^a$.
 
%   For the specific parameter setting, only event E1 is realized.
% }

% {\color{blue}
         \begin{figure}[!t]
  \centering
  \includegraphics[width=0.9\linewidth]{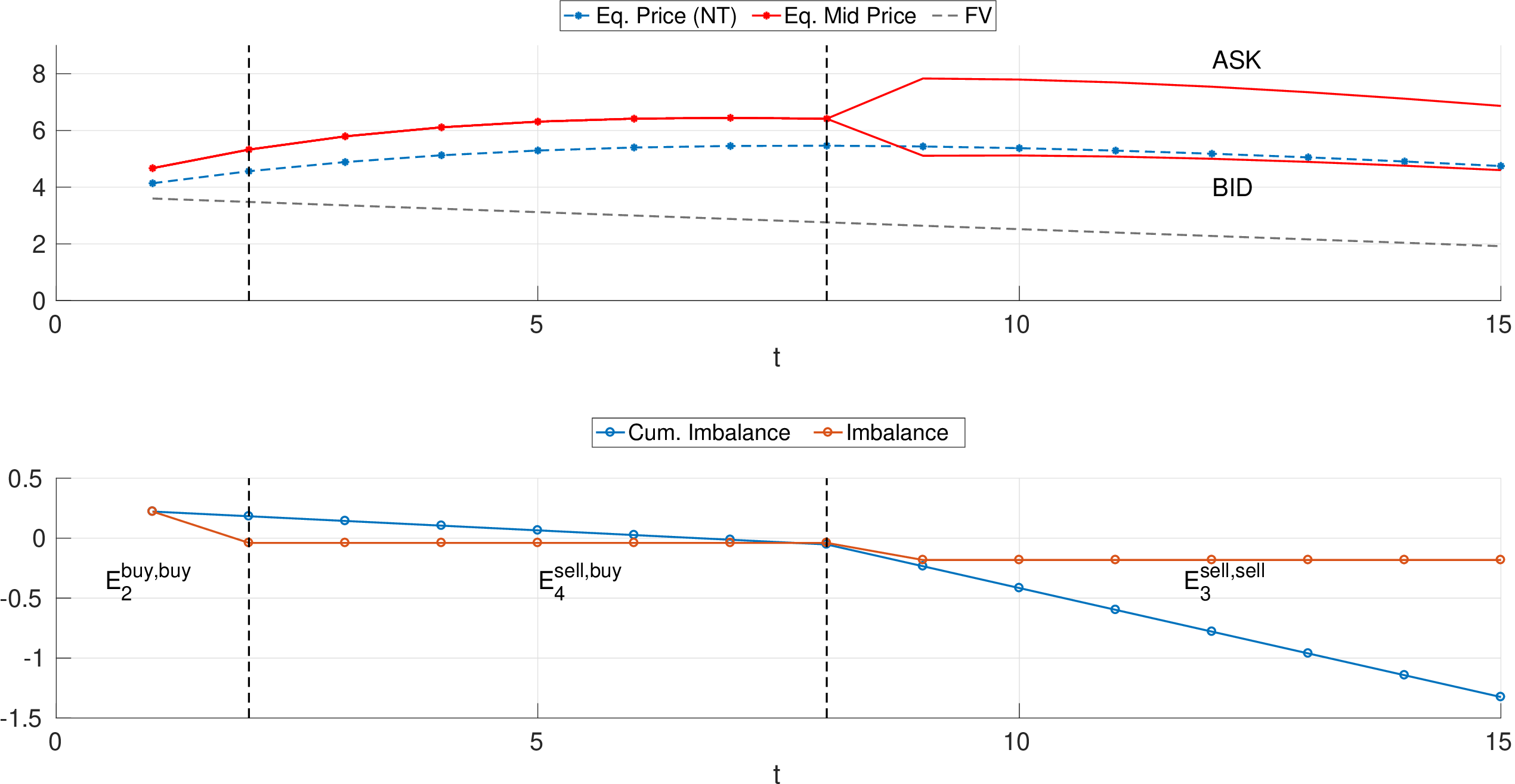}
    \caption{Equilibrium average price and order imbalance. (Top panel) Average mid-price (red) and equilibrium (blue dotted lines) price dynamics from the heterogeneous model of 
    \cite{baghestanian2015traders} where $J_F=6\%$ and $J_S=4\%$ of the total population.
    The dashed grey lines split the trading period depending on the realized event, which characterized the order imbalance (Bottom panel).
    % Red dotted lines exhibit 
    % the average bid and ask price, $p_t^b <p_t^a$.
    % We compute the current ask and bid price for each trading period, depending on the event realization. 
    }
    \label{fig:eq_price_puzzello_fund}
\end{figure}

         \begin{figure}[!t]
  \centering
  \includegraphics[width=0.9\linewidth]{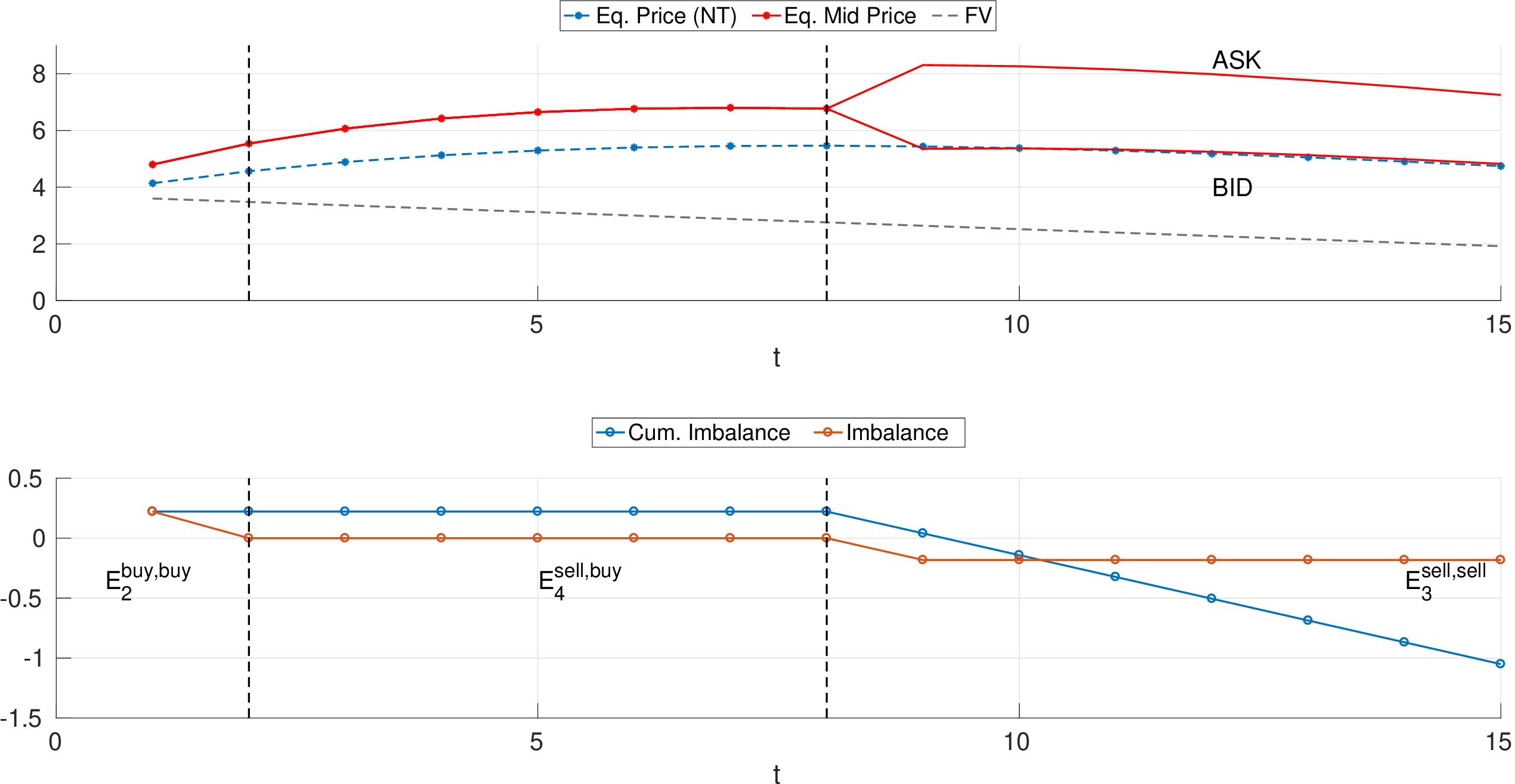}
    \caption{Equilibrium average price and order imbalance. (Top panel) Average mid-price (red) and equilibrium (blue dotted lines) price dynamics from the heterogeneous model of 
    \cite{baghestanian2015traders}  where $J_F=5\%$ and $J_S=5\%$ of the total population.
    The dashed grey lines split the trading period depending on the realized event, which characterized the order imbalance (Bottom panel).
    % Red dotted lines exhibit 
    % the average bid and ask price, $p_t^b <p_t^a$.
    % We compute the current ask and bid price for each trading period, depending on the event realization. 
    }
    \label{fig:eq_price_puzzello_eq}
\end{figure}

         \begin{figure}[!t]
  \centering
  \includegraphics[width=0.9\linewidth]{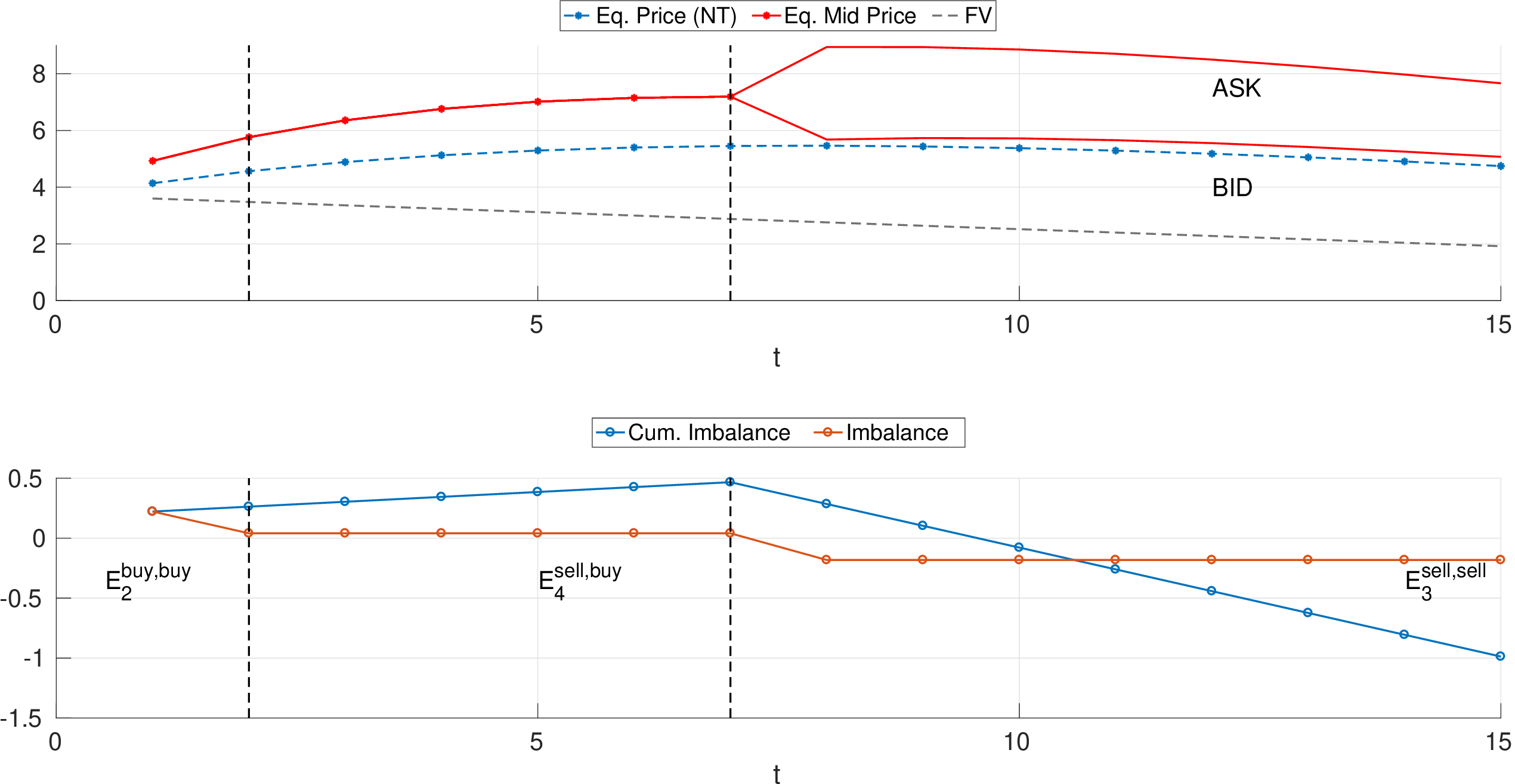}
    \caption{Equilibrium average price and order imbalance. (Top panel) Average mid-price (red) and equilibrium (blue dotted lines) price dynamics from the heterogeneous model of 
    \cite{baghestanian2015traders}  where $J_F=4\%$ and $J_S=6\%$ of the total population.
    The dashed grey lines split the trading period depending on the realized event, which characterized the order imbalance (Bottom panel).
    % Red dotted lines exhibit 
    % the average bid and ask price, $p_t^b <p_t^a$.
    % We compute the current ask and bid price for each trading period, depending on the event realization. 
    }
    \label{fig:eq_price_puzzello_spec}
\end{figure}

% }

%          \begin{figure}[!t]
% \centering
% \subfloat[][$J_F=6\%$ and $J_S=4\%$.]
% {\includegraphics[width=0.4\linewidth]{eq_price_motivation_noise_900_spec_40_fund_60.eps}} \\
% \subfloat[][$J_F=5\%$ and $J_S=5\%$.]
% {\includegraphics[width=.4\textwidth]{eq_price_motivation_noise_900_spec_50_fund_50.eps}}\\
% \subfloat[][$J_F=4\%$ and $J_S=6\%$.]
% {\includegraphics[width=0.4\linewidth]{eq_price_motivation_noise_900_spec_60_fund_40.eps}}
%   % \includegraphics[width=1\linewidth]{eq_price_motivation_noise_900_spec_50_fund_50.eps}
%    % \includegraphics[scale=0.5]{eq_price_motivation_noise_900_spec_60_fund_40.eps}
%     \caption{Equilibrium average price and order imbalance. (Top panel) Average mid-price (red) and equilibrium (blue dotted lines) price dynamics from the heterogeneous model of 
%     \cite{baghestanian2015traders}.
%     The dashed grey lines split the trading period depending on the realized event, which characterized the order imbalance (Bottom panel).
%     % Red dotted lines exhibit 
%     % the average bid and ask price, $p_t^b <p_t^a$.
%     % We compute the current ask and bid price for each trading period, depending on the event realization. 
%     }
%     \label{fig:eq_price_puzzello}
% \end{figure}

As observed also in Section \ref{sec_DU_eq_price}, 
the average equilibrium mid-price does not converge to the fundamental value, even if fundamentalists and speculators agents are included in the market.
Indeed, as explained by \cite{de1990noise}, this phenomenon may be attributed to the noise traders' risk, which discourages other rational agents from facing noise traders, causing so this significant deviation of the price from fundamental value.

 % {\color{blue}
Finally, we observe that regardless of the proportion of fundamentalists and speculators, for a very large proportion of noise traders, the price dynamics are very similar. However, as the proportion of speculators increases, we can observe how the mid-price of the asset rises, and in particular, the hump shape becomes more pronounced. This is also evident when observing the dynamics of the order imbalance (bottom panels in Figures \ref{fig:eq_price_puzzello_fund}, \ref{fig:eq_price_puzzello_eq} and \ref{fig:eq_price_puzzello_spec}). For example, in the case where $J_S > J_F$, as shown in Figure \ref{fig:eq_price_puzzello_spec}, the cumulative imbalance increases during phase $E_4$ (as the demand from speculators outweighs the supply from fundamentalists, thus increasing the aggregate cumulative imbalance).
We also note that in the case of $J_S > J_F$, phase $E_4$ terminates earlier compared to the case where $J_S < J_F$. This is also an effect of the demand from speculators that further inflates the bubble. Conversely, in the case of $J_S < J_F$, Figure \ref{fig:eq_price_puzzello_fund}, the cumulative imbalance decreases during $E_4$, as the supply from fundamentalists outweighs the demand from speculators, although it remains positive (demand skew) throughout phase $E_4$, thus alimenting the price bubble until it bursts in phase $E_3$.
% }

%  \begin{figure}[!t]
%     \centering
%     \includegraphics[width=0.9\textwidth]{eq_price_puzzello.eps}
%     \caption{Average mid-price (red) and equilibrium (blue) price dynamics from the heterogeneous model of 
%     \cite{baghestanian2015traders}. Red dotted lines exhibit 
%     the average bid and ask price, $p_t^b <p_t^a$.
%     We compute the current ask and bid price for each trading period, depending on the event realization. 
%     }
%     \label{fig:eq_price_puzzello}
% \end{figure}
 %qui inserire parametri

 \subsection{The Two-Asset Case}\label{sec_multi_scenario_puzz}
 
%  The analytical study of bid and ask price dynamics can be challenging and beyond the scope of this work.
We analyse the equilibrium price affected by different investment strategies in a multi-asset scenario.
We recall that we focus on the assets fundamental values discussed in Section \ref{sec_market_setup}, where $\overline{d}_1>0$ and $\overline{d}_2=0.$

We consider for the moment the dynamics of asset $P_1$.
For the sake of notation, we will not report the subscript $1$, since the below reasoning is valid for a generic asset with $\overline{d}\geq0$.
We restrict our analysis in the limit case of $\alpha^F=1$ and $\gamma_{1}=0$ when both fundamentalists and speculators follow one of the market factors discussed in Section \ref{sec:two_assets}. 

Therefore, $l_{t}=FV_{t}$ for all $t$ and $E_{p_{t}}=\gamma_{2} FV_{t}$ is decreasing in time, when $FV_{t}$ is decreasing.
Traders follow asset 1 and decide the position on asset 2 using the respective factor. 
For a generic asset $i$, we recall that a fundamentalist decides to buy if $l_{t,i}\leq FV_{t,i}$ and a speculator decides to sell if $E_{p_{t,i}}\geq E_{p_{t+1,i}}$.
Thus, for all $t$ fundamentalist will buy asset $P_1$ while speculator will sell it, i.e., event $E_1$ is realized for asset $P_1$. Moreover, $q_{t}^F=FV_{t}$ and $q_t^S=\frac{\gamma_2}{2}(FV_t+FV_{t+1})=
\frac{\gamma_2}{2}(FV_t+FV_{t}-\overline{d})=\gamma_2 (FV_t-\frac{1}{2}\overline{d})$.
Therefore,
\begin{equation}
\label{eq_ask_bid_two_assets}
\begin{cases}
p_t^a&=\frac{0.5 J_N q_t+J_S \gamma_2 (FV_t-\frac{1}{2}\overline{d})}{
0.5 J_N +J_F}\\
p_t^b&=\frac{0.5 J_N q_t+J_F FV_t}{
0.5 J_N +J_S}.
\end{cases}\end{equation}

 We first consider the case when
 $J_F=J_S$.
% We now proof the following results,
% under the assumption that .
\begin{pr}\label{prop_bid_ask_spread}
When fundamentalists and speculator agents form expectations for the next
 market-clearing price 
 without considering the previous trading period price, i.e., $\alpha^F=1$ and $\gamma_1=0$, then
if $J_F=J_S$, $p_t^a>p_t^b$ if and only if $\gamma_2>\frac{2 FV_t}{2FV_t-\overline{d}} \ \forall t$.
\end{pr}

Therefore, fundamentalists and speculators sustain the demand and supply regardless of noise traders.
On the other hand, the price dynamics is mainly led by noise traders who execute orders inside the spread formed by fundamentalists and speculators. On average, the mid-price $p_t$ will characterize the dynamics of the price realizations outlined by $p_t^b$ and $p_t^a$ quotes.
We may consider $p_t$ as the benchmark price realizations used by noise traders when they post their quotes $q_t$.
Therefore, under the previous assumptions, the mid-price dynamics of asset $P_1$ is given by
$$p_t=\frac{p_t^a+p_t^b}{2}
=\frac{J_N q_t+ J_F(FV_t+\gamma_2 (FV_t-\frac{\overline{d}}{2}))}{J_N+2J_F},$$
\begin{equation}
\label{eq_mid_price_puzzello}
p_t=\frac{J_N (\kappa \frac{FV_t}{2} (1-\alpha)+\alpha p_{t-1})+ J_F(FV_t+\gamma_2 (FV_t-\frac{\overline{d}}{2}))}{J_N+2J_F}.
\end{equation}

We remark that, following the same above reasoning, we may obtain the price description also for asset $P_2$, see Section \ref{sec_solving_idproblem}.

% {\color{red}
Let us now consider the general case when $J_F\neq J_S$ and let us denote $$\varrho(J_F, J_S, FV_t, \overline{d}) = \frac{2J_F FV_t}{J_S (2FV_t - \overline{d})}$$ and 
\[\varrho^{*} (J_F, J_S, J_N, FV_t, \overline{d}, q_t) = 
     \left[
     \frac{J_N + 2J_F}{J_N + 2J_S} 
     \cdot (J_N q_t + 2J_F FV_t) - J_N q_t
     \right] \cdot  \frac{1}{J_S (2FV_t -\overline{d})}.\]
We can now state the following Proposition. 

\begin{pr}\label{general_prop_F_neq_S}
When fundamentalists and speculator agents form expectations for the next
 market-clearing price 
 without considering the previous trading period price, i.e., $\alpha^F=1$ and $\gamma_1=0$, then:
 \begin{itemize}
     \item[a)] If $J_F> J_S$, then 
         $\gamma_2< \varrho(J_F, J_S) \Rightarrow p_t^a < p_t^b.$
     \item[b)] 
     If $J_F< J_S$, then 
     $\gamma_2> \varrho(J_F, J_S) \Rightarrow p_t^a > p_t^b$.
     \item[c)] $\gamma_2 = \varrho^{*} (J_F, J_S, J_N, FV_t, \overline{d}, q_t) \Leftrightarrow p_t^a = p_t^b$.
     
 \end{itemize}

\end{pr}

% qui 24 July 2024
Therefore, when $J_F < J_S$ and $\gamma_2>\varrho(J_F, J_S)$, $p_t^a>p_t^b$, meaning that we may study as above the mid price dynamics $p_t=(p^a_t +p_t^b)/2$, where the dynamics of the price will be characterized by noise traders orders executed inside the spread and $p_t$ will be determined by Equations \eqref{eq_ask_bid_two_assets}. On the other hand, when $J_F >J_S$ and $\gamma_2<\varrho(J_F, J_S)$, the spread will be fully closed since $p_t^a < p_t^b$ and we may consider $p_t^a$ as the representative price dynamics instead of the mid price, see Figure \ref{fig:eq_price_puzzello_factors_diff_J_S} at end of the next section.

\begin{os}
The conditions on $\gamma_2$ in the Proposition \ref{general_prop_F_neq_S} points a) and b) can be easily satisfied for generic values of $\gamma_2$.
Indeed,
if $J_F\gg J_S$, $\varrho(J_F, J_S, FV_t, \overline{d})$ will diverge and thus $\gamma_2< \varrho(J_F, J_S, FV_t, \overline{d})$ could be satisfied for generic values of $\gamma_2$ since $\varrho(J_F, J_S, FV_t, \overline{d})$ will be far from zero. On the other hand, when $J_F \ll J_S$ then 
$\varrho(J_F, J_S, FV_t, \overline{d})$ will decrease to 0 and so $\gamma_2>\varrho(J_F, J_S, FV_t, \overline{d})$ will be satisfied for values of $\gamma_2$ sufficiently far from zero.
\end{os}

\begin{os}\label{remark_gamma_2_time_varying}
If $FV_t$ is not constant, the condition in Proposition \ref{general_prop_F_neq_S} point c) cannot be satisfied since $\gamma_2$ is a constant, and therefore, the price will not be in equilibrium.
However, if we assume that $\gamma_{2,t}$ could be a time varying parameter for the speculators, i.e., speculators can change their weight on the fundamental value over time when forming expectations of future prices, then the equilibrium condition can be satisfied.
Although the existence of the equilibrium price could be a valid motivation to allow 
$\gamma_2$ to vary over time, this generalization of the \cite{baghestanian2015traders} model requires a valid economic justification, which should be potentially tested with appropriate experiments.
\end{os}

% }

\subsubsection{Factor Investing and Investment Strategies}
%generalizzare tutot il modello in termini di parametri

%arrivato qui 03/01/2022

We now combine the \cite{baghestanian2015traders} model and the 
two-assets generalization with market factors of Section \ref{sec:two_assets},
% The first asset pays a positive dividend, while the second one has an expected dividend equal to zero. 
where the fundamental value of the first asset is declining among the periods while the second one is constant.
We first assume that fundamentalist and speculator agents follow asset $P_1$ and 
read the signal to decide the position on the second asset. 
The signal reads by the agents differ according to the investment strategies followed by agents in the first asset. For instance, if $l_{t,1}\leq FV_{t,1}$ fundamentalists decide to buy the first asset and buy or sell the second asset depending on the selected market factors, $v_D$ or $v_M$. 
Noise traders place orders randomly for both assets.
Although fundamentalist and speculator traders place a quote on asset $P_1$ following their strategies, we have to decide what are the quotes they will post on the other asset $P_2$. Thus, we make the following 
assumption.

\begin{as}\label{as_same_quotes_both_assets}
For all traders, the quotes on the second asset follow those of noise agents, i.e., on average
$q_{t,2}=(1-\alpha_2)\frac{\kappa_2}{2} FV_{t,2}+\alpha_2 p_{t-1,2}$.
\end{as}

 Therefore, using the same argument of Section \ref{sec:two_assets} the price dynamics of $p_{t,2}$ will be of the form
$q_{t,2}+(B_t/A_t -1).$ The demand and supply imbalance depends on the market factors followed by traders and on the parameter specification of
both fundamentalists and speculators.
We assume that $\alpha^F=1$, $\gamma_{1}=0$ and we first consider the case $J_F=J_S$. Therefore, 
fundamentalists and speculators will buy and sell, respectively the first asset. 
If they select the same factor, i.e., they are both directionals or market neutrals, 
the price dynamics of asset $P_2$ is in equilibrium and
$p_{t,2}=q_{t,2}$.

\begin{te}\label{te_identificatoin_equilibria}
Under Assumptions, \ref{as_cash}, \ref{as_quotes},  \ref{as_probab_two_assets}-b, \ref{asm_pi}, \ref{as_same_quotes_both_assets}, 
in the two assets generalization, when fundamentalists and speculators decide position on the second asset selecting a market factor, they  form expectations for the next
 market-clearing price 
 without considering the previous trading period price, i.e.,  $\alpha^F=1$, $\gamma_1=0$, and $J_F=J_S$
and both fundamentalists and speculators select 
the same factor, the price dynamics of asset $P_2$ 
is in equilibrium and it is given by $p_{t,2}=q_{t,2}.$
\end{te}

In Figure \ref{fig:eq_price_puzzello_factors}
we display the average mid-price for both assets where the parameters are setting to $\kappa_1=4$, $\kappa_2=2$, $\alpha_1=\alpha_2=0.85$, $J_N=50$, $J_S=J_F=25$, $\alpha^F=1$, $\gamma_1=0$, $\gamma_2=4$ and both the fundamentalists and speculators follow the same factor (i.e., when they are both directionalists or they are both market-neutrals, since we obtain exactly the same price dynamics).
The dividend distributions are the same as described in Section \ref{sec_market_setup}. From Theorem \ref{te_identificatoin_equilibria} asset $P_2$ is in equilibrium which coincides with the fundamental value.
We display the average mid-price dynamics since
from Proposition \ref{prop_bid_ask_spread} $p^b_{t,1}<p_{t,2}^a$ and the price dynamics is characterized by noise traders orders which are executed inside the spread formed by fundamentalists and speculators.
 
%              \begin{figure}
%   \centering
%   \includegraphics[width=0.6\linewidth]{eq_price_2_assets_factor_invest.eps}
%     \caption{RD for two-asset markets. $\kappa_1 =4$, $\kappa_2 = 2$, $\alpha_1=\alpha_2 = 0.85$, 
%     $\alpha^F=1$, $\gamma_{1}=0$, $\gamma_{2}=4$ and $J_S=J_F=(1-J_N)/2$.}
% \end{figure}

 \begin{figure}[!t]
    \centering
    \includegraphics[width=1\textwidth]{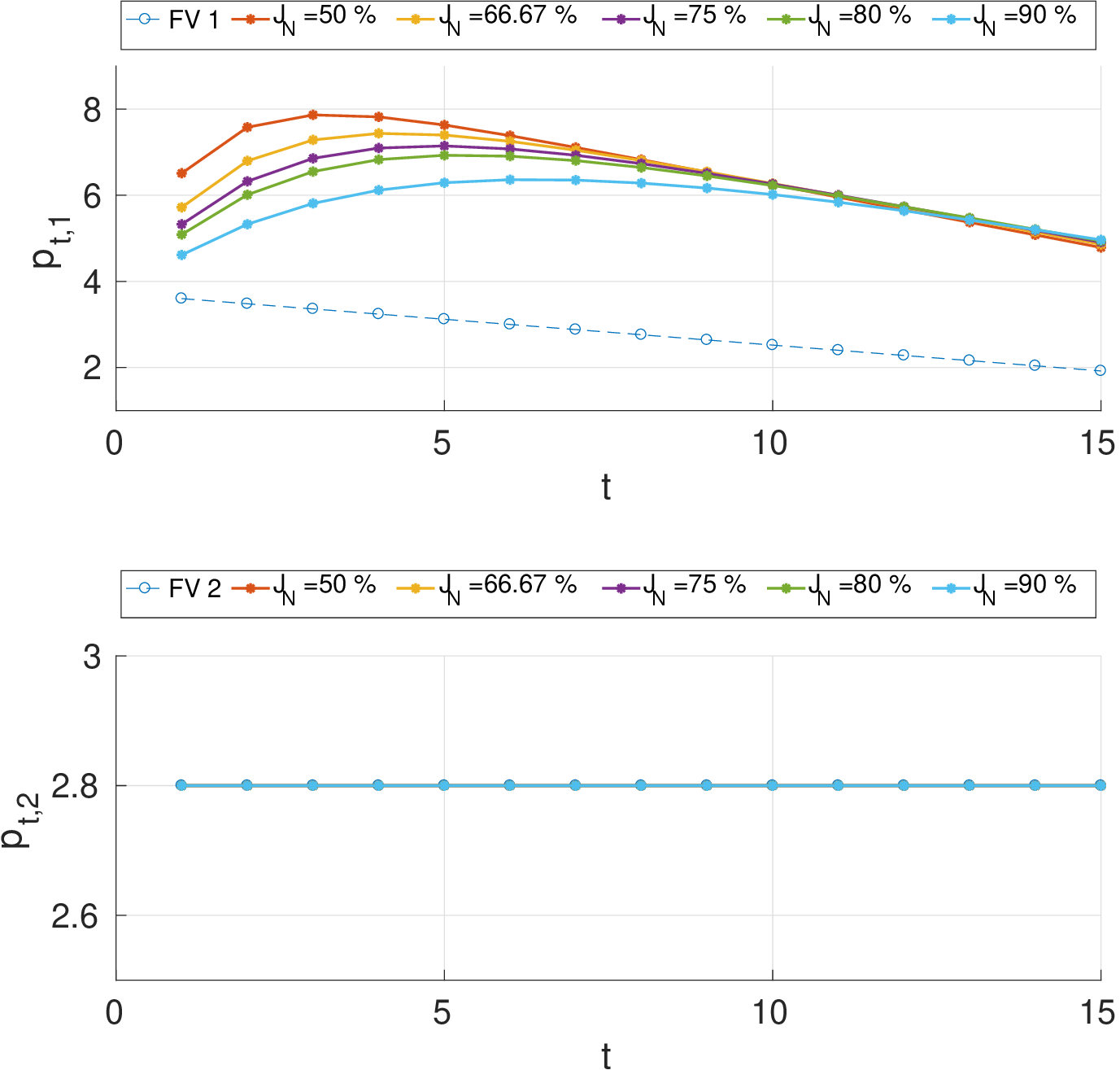}
    \caption{
    Average mid-price for two-asset markets by varying the percentage of noise traders $J_N$, from the heterogeneous model of 
    \cite{baghestanian2015traders} when both fundamentalist and speculators follow the same factor (directional or market-neutral).  $\kappa_1 =4$, $\kappa_2 = 2$, $\alpha_1=\alpha_2 = 0.85$, 
    $\alpha^F=1$, $\gamma_{1}=0$, $\gamma_{2}=4$ and $J_S=J_F=(1-J_N)/2=25$.
    For each trading period we compute the current ask and bid price, depending on the event realization. 
    The noise trader quotes are updated using the previous period mid-price, i.e., $q_t=(1-\alpha)\overline{u}_t+\alpha (p^a_{t-1}+p^b_{t-1})/2 $. 
The fundamental value is shown in blue dotted lines.}
    \label{fig:eq_price_puzzello_factors}
\end{figure}

% {\color{blue}

Moreover, the misvaluation of 
generated by the price bubble, $P_1$, decreases when the percentage of noise traders, $J_N$, increases, see also the RD measures in Figure \ref{fig:eq_price_puzzello_factors_RD}. On the other hand,
the misvaluation of asset $P_2$ is invariant from $J_N$.
However, when traders' confusion  on the fundamental value of asset $P_2$ increases, i.e., the parameter $\kappa_2$, the price is still in equilibrium but it exhibits the price bubble shape and a subsequent significant overvaluation.

             \begin{figure}[!t]
  \centering
  \includegraphics[width=1\linewidth]{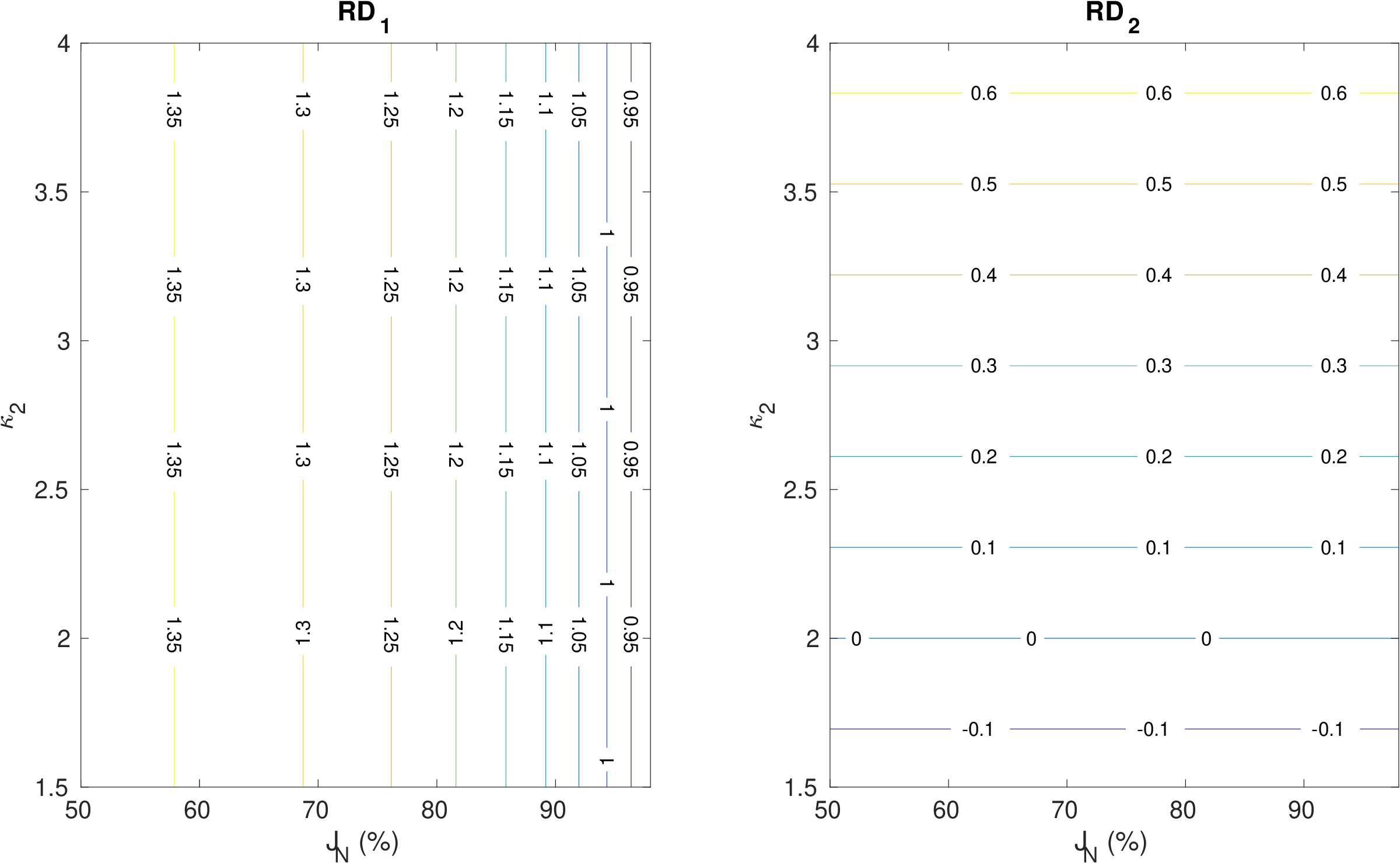}
    \caption{RD for two-asset markets when the percentage of noise traders, $J_N$, and traders' confusion about the fundamental value of asset $P_2$, i.e., $\kappa_2$, are varying. }\label{fig:eq_price_puzzello_factors_RD}
\end{figure}

Regardless the market factors, for asset $P_1$, the event $E1$ is realized, while for asset $P_2$, when traders follow the directional (market-neutral) factor, the event $E1$ ($E4$) is realized, respectively.
When fundamentalists and speculators follow opposite market factors, the price dynamics of asset 2 is no longer in equilibrium, and it is driven by the demand and supply imbalance of event $E_2$ or $E_3$, depending on if the fundamentalists/speculators are directional/market-neutral or market-neutral/directional, respectively.

The previous result
highlights an identification issue since the 
price equilibrium is reached when fundamentalists and speculators follow the directional or market-neutral factor.
% This comes out because we specify an
% This is not surprising since
% the equilibrium dynamic is achieved without a precise economic interpretation. 
Thus,
is the equilibrium characterised by the directional or market-neutral factor?
The next section will propose a possible economic motivation to identify one of the two-factor strategies characterizing the equilibrium and solving the previous identification problem. 
In particular, we will specify an investment strategy also for asset $P_2$ and
we will also relax Assumption \ref{as_same_quotes_both_assets}.
% }

% {\color{red}
Finally, we deal with the case when $J_F\neq J_S$.
In Figure \ref{fig:eq_price_puzzello_factors_diff_J_S}
we display the average mid-price\footnote{When $J_F>J_S$ we consider $p_{t,1}^a$ for asset $P_1$, see the discussion below Proposition \ref{general_prop_F_neq_S}.} for both assets by varying the percentage of fundamental and speculator traders, where $J_S+J_F=20\%$, $J_N=80\%$, the parameters are set to $\kappa_1=4$, $\kappa_2=2$, $\alpha_1=\alpha_2=0.85$, $\alpha^F=1$, $\gamma_1=0$, $\gamma_2=4$ and both the fundamentalists and speculators follow the same directional factor. 

As previously observed, for asset $P_1$, fundamentalists will post on the buy side while speculators will post on the sell side. Therefore, when $J_F>J_S$, fundamentalists will place quotes closer to the fundamental value, effectively pushing the dynamics towards $FV_t$. Conversely, when $J_F<J_S$, speculators will place sell quotes that deviate the price from its fundamental value, and the price dynamics will exhibit a more pronounced hump-shaped pattern, consistent with what was observed in Section \ref{sec_theoretical_motivation}.

For asset $P_2$, if both types of traders are directional (we recall that traders read the signal to buy or sell from asset $P_1$ and transform it according to the factor they follow for asset $P_2$\footnote{Without the assumption that both traders follow the same factor, the dynamics become intractable in order to derive general results (at least we are not capable of studying them), and therefore we postpone this extension to future research.}), and $J_F>J_S$ the price $p_{t,2}$ will deviate from its fundamental value, $p_{t,2} > FV_{t,2}$, because there will be more buyers (i.e., fundamentalists). This is because the price dynamics will be of the form $q_{t,2} + (B_t/A_t -1)$ under Assumption \ref{as_same_quotes_both_assets}, as previously observed. Conversely, if $J_F<J_S$, there will be more sellers and thus $p_{t,2} < FV_{t,2}$.
In the case where both traders follow the market-neutral strategy, the two dynamics are inverted. That is, if 
$J_F>J_S$, $p_{t,2} <FV_{t,2}$, because fundamentalists will be sellers for asset $P_2$ when they follow the market-neutral factor, and similarly if $J_F<J_S$, $p_{t,2} > FV_{t,2}$.

% posted by the agents as a proxy for the mid price dynamics.\footnote{We consider the quotes for asset 2, since we have no results showing that in the general case $J_F\neq J_S$ the equilibrium price asset 2 exists and it is equal to $q_t$. Analogously to Proposition \ref{general_prop_F_neq_S} we conjecture that when $J_F\neq J_S$ there is no equilibrium since $p_t^a>p_t^b$ or $p_t^b>p_t^a$  depending whether $J_F < J_S$ or $J_F>J_S$.}.

\begin{figure}[!t]
    \centering
    \includegraphics[width=1\textwidth]{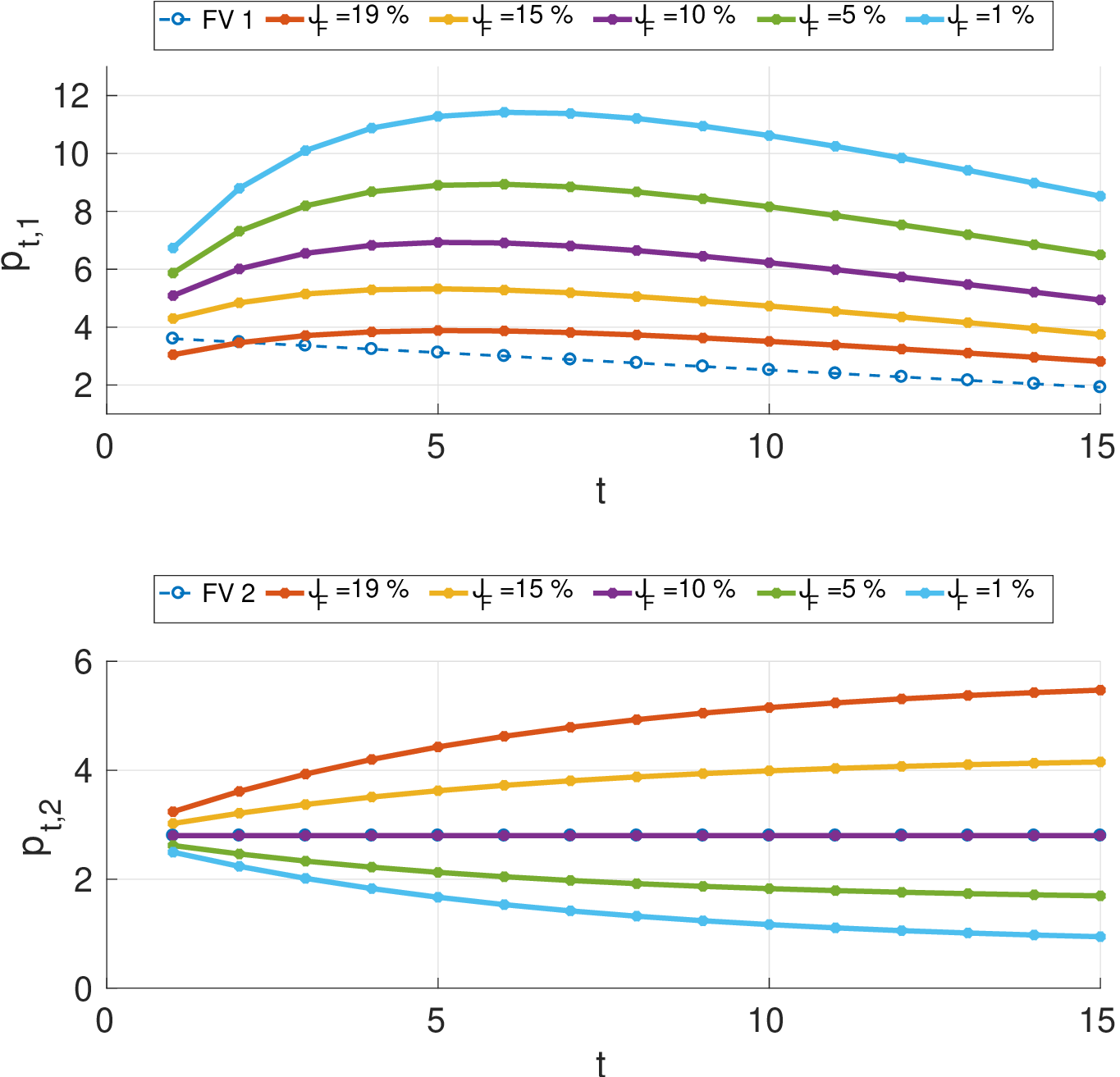}
    \caption{
    Average mid-price for two-asset markets by varying the percentage of fundamental and speculator traders, where $J_S+J_F=20\%$ and $J_N=80\%$, from the heterogeneous model of 
    \cite{baghestanian2015traders} when both fundamentalist and speculators follow the same factor (directional).  $\kappa_1 =4$, $\kappa_2 = 2$, $\alpha_1=\alpha_2 = 0.85$, 
    $\alpha^F=1$, $\gamma_{1}=0$ and $\gamma_{2}=4$.
    For each trading period we compute the current ask and bid price, depending on the event realization. 
    The noise trader quotes are updated using the previous period mid-price, i.e., $q_t=(1-\alpha)\overline{u}_t+\alpha (p^a_{t-1}+p^b_{t-1})/2 $. 
The fundamental value is shown in blue dotted lines. For $p_{t,2}$ the fundamental value dynamics coincides with the line corresponding to $J_F=10\%$.}
    \label{fig:eq_price_puzzello_factors_diff_J_S}
\end{figure}

% }

\subsubsection{Solving the identification problem of factor-investing equilibrium}\label{sec_solving_idproblem}

We now investigate a possible economic interpretation of the previous results. To do that, we have to extend the \cite{baghestanian2015traders} model
to the two-assets case without relying on market factors.

We now assume that a trader who follows a 
particular investment strategy for the first asset, e.g., the fundamentalist one, uses the same strategy also for the second asset. This assumption will replace the more constraining Assumption \ref{as_same_quotes_both_assets}.

\begin{as}\label{same_inv_strategies}
A fundamentalist (speculative) trader for asset $P_1$ is also a fundamentalist (speculative) for the second asset.
\end{as}

The parameters of investment strategies differ for the two assets, and we generalize the previous model specification using $\alpha^F_i$ as the fundamentalist anchoring parameter for asset $i$ and 
$\gamma_{1,i}$ and $\gamma_{2,i}$ as the parameters used by speculators to form expectations about next market-clearing price for asset $i$. The quote size for the two assets is trivially the average of the corresponding clearing prices expectations.

Thus, we may assume as for asset $1$ that fundamentalists and speculator agents form expectations for the next
 market-clearing price for asset $P_2$
 without considering the previous trading period price, i.e., $\alpha^F_2=1$ $\gamma_{1,2}=0$.
%commentare meglio le ipotesi
%scrivere meglio le ipotesi con numeri
\begin{co}\label{te_solving}
Under Assumptions, \ref{as_cash}, \ref{as_quotes},  \ref{as_probab_two_assets}-b, \ref{asm_pi}, \ref{same_inv_strategies},
if $\overline{d}_2=0$, 
when fundamentalists and speculator agents form expectations for the next
 market-clearing price 
 without considering the previous trading period price for both assets, i.e., $\alpha_{\cdot}^F=1$ and $\gamma_{1,\cdot}=0$, then 
the price of asset $2$ is in equilibrium 
if and only if $\gamma_{2,2}=\varrho^{*} (J_F, J_S, J_N, FV_{0,2}, 0, q_{t,2}).$
\end{co}

% if $\gamma_{2,2}=1$, from \eqref{eq_ask_bid_two_assets}, $p_{t,2}^a=\frac{0.5 J_N q_{t,2}+J_F FV_{t,2}}{0.5 J_N +J_F}=p_{t,2}^b.$
% On the other hand, if $p^a_{t,2}=p^b_{t,2}$ then from \eqref{eq_ask_bid_two_assets},
% $J_F \cdot FV_{t,2}\cdot (\gamma_{2,2}-1)=0$ and we may conclude.

Thus, following the same reasoning of the proof of Theorem \ref{te_identificatoin_equilibria}, fundamentalists decide to buy while speculators decide to sell asset $P_2$, i.e., they have the same position they have for asset $P_1$. 
This kind of demand and supply entanglement is the same as for the previous factor investing strategies model where agents follow the directional market factor.

% {\color{red}
\begin{os}
As discussed in Remark \ref{remark_gamma_2_time_varying}, the condition 
on $\gamma_{2,2}$ of Corollary \ref{te_solving} could not be satisfied if $\gamma_2$ is a constant, since $\varrho^{*}$ will vary over time with $q_{t,2}$. However,
when $J_F=J_S$,  
$\varrho^{*} (J_F, J_S, J_N, FV_{0,2}, 0, q_{t,2})=\left[
     \frac{J_N + 2J_F}{J_N + 2J_S} 
     \cdot (J_N q_t + 2J_FFV_{0,2}) - J_N q_t
     \right] \cdot  \frac{1}{2J_S FV_{0,2}}=1$, and the condition is satisfied when $\gamma_{2,2}=1.$ Therefore, in the following, we will focus on the case $J_F=J_S$.      
% where $J_F=J_S$ the quotes of both speculators and fundamentalists will be constant. Moreover, even the quotes of the noise traders will be constants\footnote{Solving recursively the expression for $q_{t,2}=(1-\alpha_2) \frac{\kappa_2}{2}FV_{t,2} + \alpha_2$ and recalling that $p_{0,2} = FV_{0,2}$,
% $q_{t,2}=(1-\alpha_2) \frac{\kappa_2}{2}FV_{0,2} + \alpha_2$
% }

\end{os}
% }
The average price dynamics generated by the model
where fundamentalist and speculators follows the same strategies for both assets (without relying on factors)
where 
$\kappa_1=4$, $\kappa_2=2$, $\alpha_1=\alpha_2=0.85$, $J_N=50$, $J_S=J_F=25$, $\alpha^F_1=\alpha^F_2=1$, $\gamma_{1,1}=\gamma_{1,2}=0$, $\gamma_{2,1}=4$ and $\gamma_{2,2}=1$ coincides exactly with those exhibited by Figure \ref{fig:eq_price_puzzello_factors}, where fundamentalists and speculators use the same factor strategies (which can be directional or market-neutral) for asset $P_2$.
Moreover,
under the model specification considered above, fundamentalists and speculators post the same quotes, i.e., $q_t^F=q_t^S=FV_t$.
The price dynamics of asset $P_2$ is characterized
by a weighted average of the quote of noise traders and that of fundamentalists and speculators, i.e., from Equation \eqref{eq_mid_price_puzzello} and recalling that $FV_{t,2}$ is constant among trading periods we obtain that
\[
p_{t,2}=\frac{J_N 
(FV_{t,2}\frac{ \kappa_2}{2} (1-\alpha_2)+\alpha_2 p_{t-1,2}  )+
2 J_F FV_{t,2}}{
J_N+2J_F}
\]
Therefore, when noise traders have no confusion on $FV_2$, i.e., $\kappa_2=2$, since $p_{0,2}=FV_{1,2}=FV_{t,2}$ we may relate the equilibrium found
in Theorem \ref{te_identificatoin_equilibria} to the previous one. 

 In other words, by extending the \cite{baghestanian2015traders} in the two-asset case, we can identify the two-assets equilibrium described by Theorem \ref{te_identificatoin_equilibria}, see Figure \ref{fig:graph_caption}.
     The two-asset equilibrium described in Theorem \ref{te_identificatoin_equilibria} can be reached by two paths. By assuming the model where fundamentalists/speculators follow the market-neutral factor for posting their quotes on asset $P_2$, model $\mathscr{M}^{spec,mn}_{fund,mn}$, or by assuming that they follow the directional, model 
     $\mathscr{M}^{spec,direc}_{fund,direc}$.
     However, from Corollary \ref{te_solving} we know that when both fundamentalists and speculators follow the same investments strategies for both assets, i.e., they are fundamentalists and speculators also for asset $P_2$, respectively, model $\mathscr{M}^{spec,spec}_{fund,fund}$, we generate the same order imbalance between demand and supply obtained by model $\mathscr{M}^{spec,direc}_{fund,direc}$. Indeed, for both assets event E1 is realized and we obtain exactly the same two-asset equilibrium and order imbalance.
 
 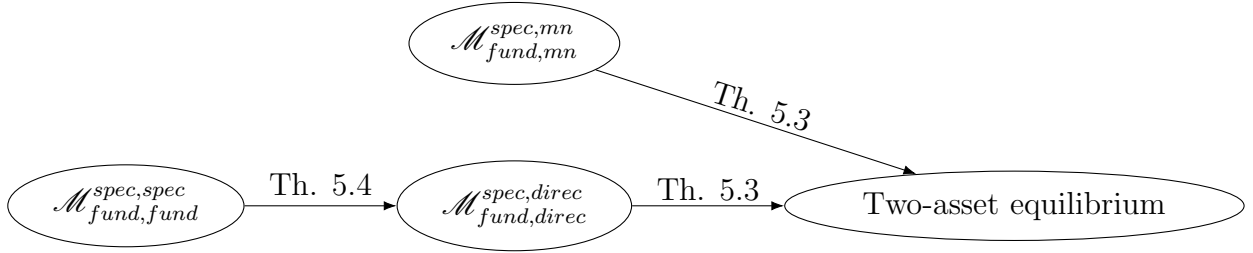
\begin{figure}[!t]
     \centering
\begin{tikzpicture}
    \node[state] (1) {$\mathscr{M}^{spec,spec}_{fund,fund}$};
    \node[state] (2) [right =2cm of 1] {$\mathscr{M}^{spec,direc}_{fund,direc}$};
    \node[state] (3) [right =2cm of 2] {Two-asset equilibrium};
    \node[state] (4) [above =of 2] {$\mathscr{M}^{spec,mn}_{fund,mn}$};

    \path (1) edge node[above] {Co. \ref{te_solving}} (2);
    \path (2) edge node[el,above] {Th. \ref{te_identificatoin_equilibria}} (3);
    \path (4) edge node[el,above] {Th. \ref{te_identificatoin_equilibria}} (3);
    % \path[bidirected] (2) edge[bend left=60] node[above] {$\epsilon_{xy}$} (3);
    %   \path (4) edge node[el,above] {$\lambda_{wz}$} (1);
    % \path[bidirected] (4) edge[bend left=50] node[el,above] {$\epsilon_{wy}$} (3);
\end{tikzpicture}
     \caption{An economic motivation of the two-assets equilibrium.
 $\mathscr{M}^{spec,spec}_{fund,fund}$ stands for the model when both fundamentalists and speculators follow the same investments strategies for both assets, i.e., they are fundamentalists and speculators also for asset $P_2$, respectively.  $\mathscr{M}^{spec,direc}_{fund,direc}$ 
 ($\mathscr{M}^{spec,mn}_{fund,mn}$) stands for
 the model where fundamentalists/speculators follow the directional (market-neutral) factor for posting their quotes on asset $P_2$.
     }
     \label{fig:graph_caption}
 \end{figure}

% {\color{blue}
Finally, we observe that when the price of asset $P_2$ is not in equilibrium, $\gamma_{2,2}\neq1$, the overvaluation measured by RD decreases when the percentage of noise traders $J_N$ increases as observed for asset $P_1$, see Figure \ref{RD_gamma_22_surf}. On the other hand, the overvaluation of the price bubble of asset $P_1$ remains invariant when we vary the 
speculators' perception (confusion) about the fundamental value of asset $P_2$, i.e., $\gamma_{2,2}$.

             \begin{figure}[!t]
  \centering
  \includegraphics[width=1\linewidth]{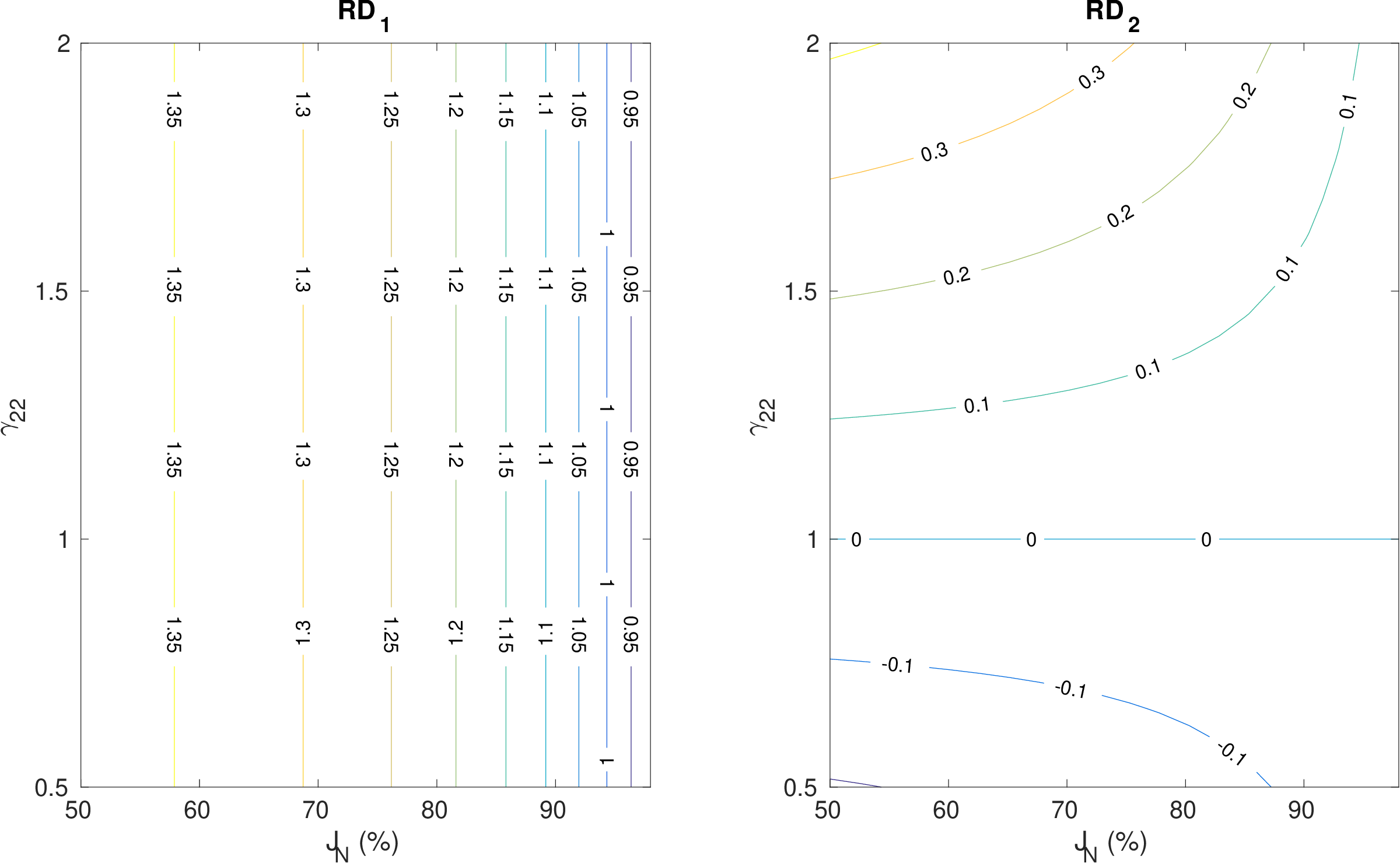}
    \caption{RD for two-asset markets when the percentage of noise traders, $J_N$, and speculators' perception about the fundamental value of asset $P_2$, i.e., $\gamma_{2,2}$, are varying. The other parameters are set to $\kappa_1 =4$, $\kappa_2 = 2$, $\alpha_1=\alpha_2 = 0.85$, 
    $\alpha^F_{\cdot}=1$, $\gamma_{1,\cdot}=0$, $\gamma_{2,1}=4$ and $J_S=J_F=(1-J_N)/2$.}\label{RD_gamma_22_surf}
\end{figure}
% }

In conclusion, if in Theorem \ref{te_identificatoin_equilibria} we have highlighted an identification issue due to the possibility of reaching equilibrium with two different factor investing strategies, now this identification problem is resolved. 
 Even if the fundamental value dynamics are different, the price of asset $P_2$ reaches an equilibrium when agents follow the same investment strategy for both assets.
Then, the equilibrium is reached using the directional market factor strategy since fundamentalists will also buy asset $P_2$ while speculators will sell it, following the same demand and supply imbalance of asset $P_1$.

\section{Conclusion}\label{sec_conc}
This work shows the existence 
of price equilibria for various agent-based models
 to investigate the origin of the typical price-bubble mechanism observed in experimental asset markets.
The equilibrium prices dynamics exhibit price-bubbles shape for those assets with a positive average dividend consistently with the experimental asset literature, e.g., \cite{smith1988bubbles}, \cite{caginalp2002speculative}, \cite{kirchler2012thar}.
When the market is not at equilibrium, 
a sharp decline in the price-bubble is observed at the end of market session, which triggers a 
price deviation from the fundamental value of the other asset.
This contagion/misvaluation effect is also displayed in the experiments of \cite{caginalp2002speculative}, where price bubbles tend to increase the volatility of other assets, and in the simulation 
results of \cite{cordoni_2021_simulation}, where the price bubble triggers asymmetric cross-impact effects.

Starting from the homogeneous DU agent-based model, we show how the price equilibrium is characterized by the so-called \emph{weak-foresight assumption}. 
Our analysis is then extended to the two-asset case discussing how the equilibrium can be reached in the presence of heterogeneous agents when factor-investing and investment strategies are introduced. 
We have shown necessary and sufficient conditions 
for which the price dynamics exhibits average bubble-crash patterns typically observed in experimental economics. The analytical expression, from which the average price dynamics for both assets can be recovered, is also derived.

We have highlighted how, under generic assumptions, the equilibrium in the two-asset extension can be reached in two alternative factor investing trading strategies, generating an identification problem. However, by extending the model 
of \cite{baghestanian2015traders} in a two-asset market, this identification issue is solved, finding motivation for describing how the equilibrium can be reached.

Our work can be extended in many directions. We 
could consider the multi-asset extension (with more than two assets) or consider different market participants as market-maker agents and study their impact on the equilibrium price dynamics.
Moreover, through market experiments, we could validate agent-based models considered and study the causes and effects of how particular dynamics might arise in a laboratory asset market. We are currently developing these experiments involving humans (professionals and students) and artificial agents in upcoming works. 

The presented results might be helpful to
experimental design and hypotheses formulation. For instance, by employing one of the model specifications, we might figure out whether, on average, price bubbles will occur or not in a determined market setting. Therefore, an experiment may be calibrated to prevent the bubble-crash pattern by exploiting our average equilibrium price dynamics analyses.

% {\color{blue}
Finally, since this work does not consider the possibility of co-evolutionary dynamics within the population of traders, as emphasized in Remark \ref{remark_evolution}, an interesting potential extension for future research involves studying how the various populations of agents co-evolve under different market conditions. For example, following the approach of \cite{bottazzi2014evolution}, it may be possible to demonstrate the existence of selection equilibria where asset prices align with their fundamental values over the long run. We postpone the exploration of evolutionary aspects of trading behaviour to future research.

%metodologia per validare i modelli in equilibrio usando le dinamiche dei prezzi comparandole con quelli realizzate nell'esperimento
%il fatto di aver caratterizzato l'equilibrio in termini di strategie (parametri) permette quindi di poter validare i modelli qui spiegati con dati esperimentali studiando l'equilibrio

%model parameters
%usare ipotesi precedenti
% nel caso di gamma_22=1 la p_t2 diventa semplicemente p_t=FV_t come nel caso direc/direc o mn/mn

% identification problem if they are both market-neutral or directional, there is no ident problem, because we are in E1
% Therefore, 
% agents follows intrinsically the direciontal factor when they use the same strategy for both assets, even if the price dynamics is different.

 %fare discorso con due asset con signal e factor investing

%capire bene i plot e specifiche dei parametri cosa controllano gli agenti cosa tracciano

% le quote sizes di ogni tipo di agente
%studiare i casi limiti di 
% $\gamma_1$ o $\gamma_2=0$
 % studiare i casi in cui speculators conoscano q_t

%in simulation p_{0} è FV_0

% $A_t=(1-\pi_t) J_N$ and $B_t=\pi_t J_N$.

% \begin{te}
% Under the previous assumptions,
% there exists an average equilibrium price if and only if 
% $\pi_t=0.5$. Moreover, 
% $$
% \overline{p}_t=q_t+q_t^F\frac{J_F}{2\cdot J_N}+q_t^S \frac{J_S}{2\cdot J_N}.
% $$

% \end{te}
% \begin{proof}
% $p_t^a=p_t^b$ if and only if
% Fare i 4 casi in cui 
% \end{proof}
 
\section*{Acknowledgements}
  The author
acknowledges the financial support from the PRIN grant no. 20177FX2A7 of the Italian Ministry of University and Research,
``How good is your model? Empirical evaluation and validation of quantitative models in economics'' and of the Leverhulme Trust Grant Award RPG-2021-359 - ``Information Content and Dissemination in High-Frequency Trading". The author gratefully
acknowledges Giovanni Cespa and Caterina Giannetti for the helpful discussions and suggestions that have helped to improve the paper.
The author is grateful to participants of the 46th Annual Meeting of the AMASES, Palermo, September 22-24, 2022 and to the anonymous referees for providing insightful and valuable comments.

\section*{Declarations of interest}
The author did not receive support from any organization for the submitted work.
 The contents of this article are reflective and express the personal opinions of the author, who is entirely independent and free from conflicts of interest related to any professional work the author may engage in.
 
\bibliographystyle{apalike}
    \bibliography{bib}
%\clearpage    
\begin{appendices}
\titlelabel{\appendixname\ \thetitle.\quad}
\appendixtitleon

\section{Proofs of the results.}
\label{sec_app_proof}
\begin{proof}[Proof of Lemma \ref{lemma_cash_endowment}]
We first observe that the random quantity $u_t$ satisfies $u_t \leq \kappa FV_t < \kappa FV_0$, since $FV_t$ is decreasing over time.
Then, since the quote $q_t$  is a weighted average of the previous trading price, $\overline{p}_{t-1}$ and $u_t$, where $\overline{p}_{0} = FV_0$, we can easily conclude by induction that $q_t < \kappa FV_0$. Indeed, since $\kappa>1$,
$q_1 < (1-\alpha) \kappa FV_0 + \alpha FV_0 < \kappa FV_0 $. Then, we observe that since the ask quotes are equal to $a_t=q_t$, by definition $\overline{p}_t\leq \max_j q_t^j$, where $q_t^j$ is the realization of $q_t$ for the $j$-th trader. Thus, if the inequality is satisfied for $t-1$ and let $j_0 = \mbox{argmax}_{j} q_{t-1}^j$, then, $q_t< (1-\alpha) \kappa FV_0 + \alpha \overline{p}_{t-1}^j\leq (1-\alpha) \kappa FV_0 + \alpha q_{t-1}^{j_0} < \kappa FV_0$. Therefore, traders can submit at least one buy order at the bid price $q_t$ for each trading period without going bankrupt, if they are endowed with the maximum possible quote for each trading period, i.e., $x_0=\kappa FV_0 T$. Obviously, this value does not represent the minimum amount of cash endowment to ensure that condition.
\end{proof}

\begin{proof}[Proof of Theorem \ref{te_1}]
By definition the market clearing price exists if and only if $p_t^b=p_t^a.$ Thus,
\[
p_t^b=p_t^a \iff \frac{\pi_t}{1-\pi_t}=\frac{1-\pi_t}{ \pi_t} 
\iff \pi_t=\frac{1}{2}.
\]
Moreover, $A_t=B_t=\frac{N}{2} \iff \pi_t=\frac{1}{2}$, and $p_t^b=p_t^a=q_t= \overline{p}_t $.
\end{proof}

\begin{proof}[Proof of Proposition \ref{pr_asset1_two_assets}]
 If $J_D=J_N$, then $p_{t,2}^a=p_{t,2}^b$ for all $t$ if and only if
 $$
   \frac{
  J_N \cdot (1-\pi_{t,2})+J_D
    }{ J_N\cdot  \pi_{t,2}+J_D 
     }=\frac{
    J_N\cdot  \pi_{t,2}+J_D 
      }{ J_N \cdot (1-\pi_{t,2})+J_D }
      \iff 
 $$
 
 $$
       J_N^2 \cdot (1+\pi^2_{t,2}-2\pi_{t,2})
      +J_D^2+2J_N J_D \cdot (1-\pi_{t,2}) =
      J_N^2\cdot  \pi^2_{t,2}+J_D^2+2 J_NJ_D \cdot \pi_{t,2}
 $$
 
 $$
 \iff
 J_N^2+2J_N J_D= 2(J_N^2+2J_N J_D)\pi_{t,2}
 \iff \pi_{t,2}=0.5.
 $$
 Furthermore, 
 $p_{t,2}^b=p_{t,2}^a=q_{2,t}$.
 \end{proof}

\begin{proof}[Proof of Theorem \ref{te_asset2}]
If $\pi_{t,2}=0.5$, then $p_{t,2}^a=p_{t,2}^b$ for all $t$ if and only if
 $$
   \frac{
  J_N \cdot 0.5+J_D (1-\pi_{t,1})+J_{MN}\pi_{t,1}
    }{ J_N \cdot 0.5+J_D \pi_{t,1}+J_{MN}(1-\pi_{t,1})
     }=\frac{
    J_N \cdot 0.5+J_D \pi_{t,1}+J_{MN}(1-\pi_{t,1})
    }{ J_N \cdot 0.5+J_D (1-\pi_{t,1})+J_{MN}\pi_{t,1} }\iff
 $$
 
 $$
 \iff
    J_N \cdot 0.5  +J_D (1-\pi_{t,1})+J_{MN}\pi_{t,1}= J_N \cdot 0.5+J_D \pi_{t,1}+J_{MN}(1-\pi_{t,1}),
 $$
  since $J_N \cdot 0.5+J_D (1-\pi_{t,1})+J_{MN}\pi_{t,1}>0$ and 
 $J_N \cdot 0.5+J_D \pi_{t,1}+J_{MN}(1-\pi_{t,1})>0$ for all $t$.
 Therefore, 
 $p_{t,2}^a=p_{t,2}^b$ if and only if 
 $$
 2\pi_{t,1} (J_{MN} -J_D)=(J_{MN} -J_D)\iff
\pi_{t,1}=1/2.
 $$
\end{proof}

\begin{proof}[Proof of Proposition \ref{pr_average_price_NOISE}]
For each event, E1, E2, E3 and E4, if $(J_F+J_S)< \infty$, when $J_N \to \infty$, the prevailing bid and ask prices will converge both to $q_t=\overline{p}_t^{Hom}$. Precisely, the spread will converge to zero, i.e.,
$\lim_{J_N \to \infty} (p_t^a-p_t^b) \to 0$. Therefore, in the limit when $J_N \to \infty$, $p_t^a \to p_t^b$, so that the market will be in equilibrium where the equilibrium market-clearing price will be equal to $\overline{p}_t^{Het} = p_t^a = p_t^b = q_t = \overline{p}_t^{Hom}.$
\end{proof}

\begin{proof}[Proof of Proposition \ref{prop_bid_ask_spread}]
If $J_F=J_S$,
$p_t^a>p_t^b$ if and only if 
$\gamma_2 (FV_t-\overline{d} /2)>FV_t$, i.e., 
$\gamma_2>\frac{2 FV_t}{2FV_t-\overline{d}}.$
\end{proof}

% {\color{red}
\begin{proof}[Proof of Proposition \ref{general_prop_F_neq_S}]
\begin{itemize}
    \item[a)] From \eqref{eq_ask_bid_two_assets} $p_t^a=\frac{ J_N q_t+J_S \gamma_2 (2FV_t-\overline{d})}{
 J_N +2J_F}$ and $
p_t^b=\frac{ J_N q_t+2J_F FV_t}{
 J_N +2J_S}$. Then, 
 \[
 p_t^b - p_t^a =
 \frac{ (J_N +2J_F) ( J_N q_t+2J_F FV_t)  - 
 (J_N +2J_S) \cdot [J_N q_t+J_S \gamma_2 (2FV_t-\overline{d})] }{(J_N +2J_F) \cdot(J_N +2J_S)}
 \]
 but since $J_F>J_S$, $(J_N +2J_F) > (J_N +2J_S)$, so
 \[
 p_t^b-p_t^a>
 \frac{ (J_N +2J_S) [  \cancel{J_N q_t} +2J_F FV_t  - 
 \cancel{J_N q_t} - J_S \gamma_2 (2FV_t-\overline{d})] }{(J_N +2J_F) \cdot(J_N +2J_S)}
 \]
 \[
  p_t^b-p_t^a>
 \frac{ 2J_F FV_t   - J_S \gamma_2 (2FV_t-\overline{d})}{(J_N +2J_F)}.
 \]
 Then we observe that the RHS of the above inequality is positive
if $\gamma_2 <  \frac{2J_F FV_t}{J_S (2FV_t-\overline{d})} = \varrho(J_F, J_S, FV_t, \overline{d})$. 

\item[b)] Analogously to point a)
 \[p_t^a - p_t^b =
 \frac{  
 (J_N +2J_S) \cdot [J_N q_t+J_S \gamma_2 (2FV_t-\overline{d})]  -  (J_N +2J_F) ( J_N q_t+2J_F FV_t)  }{(J_N +2J_F) \cdot(J_N +2J_S)}
 \]
 but since $J_S>J_F$, $(J_N +2J_S) > (J_N +2J_F)$, so
 \[
 p_t^a-p_t^b>
 \frac{(J_N +2J_F) [  
 \cancel{J_N q_t} + J_S \gamma_2 (2FV_t-\overline{d})  - \cancel{J_N q_t} -2J_F FV_t ]}{(J_N +2J_F) \cdot(J_N +2J_S)}
 \]

 \[
 p_t^a - p_t^b >  \frac{ J_S \gamma_2 (2FV_t-\overline{d}) -2J_F FV_t }{(J_N +2J_F)}
 \]
and the RHS is positive if 
$\gamma_2 >  \frac{2J_F FV_t}{J_S (2FV_t-\overline{d})} = \varrho(J_F, J_S, FV_t, \overline{d})$.

\item[c)]

\[
p_t^a = p_t^b 
\iff 
\frac{ J_N q_t+J_S \gamma_2 (2FV_t-\overline{d})}{
 J_N +2J_F} = \frac{ J_N q_t+2J_F FV_t}{
 J_N +2J_S} \iff
\]
\[
\iff
 J_N q_t+J_S \gamma_2 (2FV_t-\overline{d}) = \frac{ (J_N +2J_F) (J_N q_t+2J_F FV_t)}{
 J_N +2J_S} \iff
\]
\[
\iff
 \gamma_2  = \left[
 \frac{ (J_N +2J_F) (J_N q_t+2J_F FV_t)}{
 J_N +2J_S} - J_N q_t\right]\frac{1}{ J_S (2FV_t-\overline{d})}\iff
\]
 \[
 \iff
\gamma_2 =\varrho^{*} (J_F, J_S, J_N, FV_t, \overline{d}, q_t)
\]

\end{itemize}  
\end{proof}
% }

\begin{proof}[Proof of Theorem \ref{te_identificatoin_equilibria}]
For asset $P_1$ we are in event $E_1$ for all trading periods, since
fundamentalist traders buy, and speculators sell.
Indeed, $l_{t,1}=FV_{t,1}$ then fundamentalists decide to buy.
On the other hand, $E_{p_{t,1}}=\gamma_2 FV_{t,1}$, where $\gamma_2>0$, is a decreasing function of time, and so 
the speculator will sell.
Therefore, when both fundamentalists and speculators are directional traders
also for asset 2 event $E_1$ is realized for all $t$, while when both are market-neutral
event $E_4$ is realized. The demand and supply imbalance vanishes in both cases since for Assumption \ref{as_same_quotes_both_assets} all the traders post the same quote $q_{t,2}$. Since $J_S=J_F$, $p_{t,2}^a=p_{t,2}^b$
and $p_{t,2}=q_{t,2}.$
\end{proof}

% {\color{red} 
\begin{proof}[Proof of Corollary \ref{te_solving}]
From Proposition \ref{general_prop_F_neq_S},
$p^a_{t,2}=p^b_{t,2}$ if and only if $$\gamma_{2,2}=\varrho^{*} (J_F, J_S, J_N, FV_{t,2}, \overline{d}_2, q_{t,2}).$$
Since $\overline{d}_2=0$, $FV_{t,2}$ is constant and we may conclude.
\end{proof}

% }

\end{appendices}

\end{document}